\documentclass[lettersize,onecolumn]{IEEEtran}
\usepackage{xcolor}
\usepackage{amsmath,amssymb,amsfonts}
\usepackage{algorithmic}
\usepackage{algorithm}
\usepackage{array}
\usepackage[caption=false,font=normalsize,labelfont=sf,textfont=sf]{subfig}
\usepackage{textcomp}
\usepackage{stfloats}
\usepackage{url}
\usepackage{verbatim}
\usepackage{graphicx}
\usepackage{cite}
\usepackage{amsthm}
\usepackage{subfig}
\usepackage{booktabs}
\usepackage{bm}
\usepackage{multirow}
\usepackage{array}
\newtheorem{theorem}{\bf Theorem}

\newtheorem{definition}{\bf Definition}
\newtheorem{remark}{\bf Remark}

\newcommand{\cM}{\mathcal{M}}
\newcommand{\sTe}{\mathsf{T_{\text{e}}}}
\newcommand{\sN}{\mathsf{N}}
\newcommand{\sF}{\mathsf{F}}
\newcommand{\sB}{\mathsf{B}}
\newcommand{\sM}{\mathsf{M}}
\newcommand{\sQ}{\mathsf{Q}}
\newcommand{\sTnc}{\mathsf{T_{\text{nc}}}}
\newcommand{\sT}{\mathsf{T}}
\newcommand{\bl}[1]{\textcolor{black}{$#1$}}

\newcommand{\tn}[1]{\textnormal{#1}}

\usepackage[normalem]{ulem}

\begin{document}

\title{On the Optimality of Coded Distributed Computing for Ring Networks}
%On the Optimality of All-to-All Broadcast Over Cache-Aided Ring Networks

\author{Zhenhao Huang, Minquan Cheng, Kai Wan, Qifu Tyler Sun, and Youlong Wu
        \thanks{ Part of this work will be presented at the IEEE International Symposium on
        Information Theory (ISIT), 2025 \cite{huang2025optimality}.}
	\thanks{Z. Huang and Y. Wu are with the School of Information Science and Technology,	ShanghaiTech University, Shanghai, China. Emails: \{huangzhh, wuyl1\}@shanghaitech.edu.cn.}
	\thanks{M. Cheng is with the School of Computer Science and Engineering, Guangxi Normal University, Guilin, China. Email: chengqinshi@hotmail.com.}
        \thanks{ K. Wan is with the School of Electronic Information and Communications, Huazhong University of Science and Technology,
        430074 Wuhan, China. Email: kai\_wan@hust.edu.cn.}
	\thanks{Q. T. Sun is with the Department of Communication Engineering, University of Science and Technology Beijing, Beijing, China. Email: qfsun@ustb.edu.cn.}
}

\maketitle

\begin{abstract}
        We consider a coded distributed computing problem in a ring-based communication network, where $\sN$ computing nodes are arranged in a ring topology and each node can only communicate with its neighbors within a constant distance $d$. 
        To mitigate the communication bottleneck in exchanging intermediate values, we propose new coded distributed computing schemes for the ring-based network that exploit both ring topology and redundant computation (i.e., each map function is computed by $r$ nodes). Two typical cases are considered:  \emph{all-gather} where each node requires all intermediate values mapped from all input files, and  \emph{all-to-all} where each node requires a distinct set of intermediate values from other nodes.
        For the all-gather case, we propose a new coded scheme based on successive reverse carpooling, where nodes transmit every encoded packet containing two messages traveling in opposite directions along the same path. Theoretical converse proof shows that our scheme achieves the optimal tradeoff between communication load, computation load $r$, and broadcast distance $d$ when $\sN\gg d$. For the all-to-all case, instead of simply repeating our all-gather scheme, we delicately deliver intermediate values based on their proximity to intended nodes to reduce unnecessary transmissions. We derive an information-theoretic lower bound for {arbitrary file placements} and show that our scheme is asymptotically optimal under the cyclic placement when $\sN\gg r$. The optimality results indicate that in ring-based networks,  the redundant computation $r$ only leads to an additive gain in reducing communication load while the broadcast distance $d$ contributes to a multiplicative gain. 
        
        %that exploits both redundancy in computation and multicast opportunity
        
\end{abstract}


\section{Introduction}

Distributed computing frameworks have emerged as a powerful paradigm for processing large-scale data and complex computational tasks, offering significant advantages over centralized approaches through efficient utilization of distributed storage and computing resources \cite{lindsay2021evolution}. This paradigm has found successful applications across diverse domains, particularly in big healthcare data analytics \cite{lee2017big}, deep learning systems \cite{zhang2018quick}, and edge computing \cite{hirsch2018augmenting}.
Nevertheless, when operating under constrained communication resources, these systems frequently encounter severe performance bottlenecks due to the substantial communication overhead of exchanging massive datasets or intermediate computational results \cite{zhang2013performance,ng2021comprehensive}. 

Recent research has demonstrated that integrating computational redundancy and coding techniques can significantly mitigate communication latency \cite{ng2021comprehensive,paschos2018role}. Notably, \cite{li2017fundamental} developed a theoretically optimal coded distributed computing scheme within the MapReduce framework, establishing the fundamental computation-communication trade-off. This seminal contribution has inspired extensive follow-up work in distributed computing. To optimize the trade-off between communication and computation latency, \cite{yu2017optimally} investigated the optimal resource allocation for coded distributed computing. Subsequent work has employed combinatorial designs to reduce system complexity by minimizing the number of sub-files and output functions, a crucial advancement for practical implementations \cite{yan2017placement,cheng2024asymptotically}.
Besides, recognizing the significant potential of wireless computing, several studies have extended coded distributed computing to wireless scenarios, incorporating various channel models such as wireless orthogonal channels \cite{li2017scalable}, wireless interference channels \cite{li2019wireless,bi2024normalized}.
These works typically assumed that the computing nodes in the system can directly communicate with each other, which may be a stringent condition in practical deployments.
In light of this, \cite{wan2020topological} developed an innovative scheme leveraging cost-effective $t$-ary fat-tree networks for MapReduce applications. A hierarchical system in which computing nodes connect to a server via relay nodes was also investigated in \cite{hu2024exploiting}.

Inspired by the practical implementation of distributed computing systems, this paper considers a coded computing system where the computing nodes are connected through a ring topology network. Ring networks provide simplified network management and efficient bandwidth utilization, facilitating their prevalence in distributed computing applications.
For example, Baidu introduced the Ring All-Reduce algorithm into deep learning, arranging computing nodes (GPUs) in a ring topology where each node communicates only with its two neighbors \cite{githubGitHubBaiduresearchbaiduallreduce}.   This architecture enables full bandwidth utilization for parallelized aggregation acceleration.  The approach has been successfully extended to hierarchical, 2D-Torus, and 3D-Torus topologies, etc \cite{jia2018highly,jiang20202d} to accommodate varying computational scales.  
Similar topological advantages have been demonstrated in federated learning systems, where ring-based algorithms not only reduce communication overhead but also enhance robustness against heterogeneous data distributions \cite{shen2023ringsfl} and Byzantine attacks  \cite{elkordy2022basil}.
The ring topology is also prevalent in satellite communication. For example, satellites in polar orbits are spaced apart by a certain angular distance, and their orbits are circular \cite{radhakrishnan2016survey, ekici2001distributed}. 
Note that in the ring network architecture, each node simultaneously operates as a source, sink, and relay node while maintaining only limited connectivity with neighboring nodes. Existing coded computing schemes, reliant on broadcast over a shared link connecting all nodes, {\it do not} apply to our problem. Therefore, it is essential to design novel schemes jointly considering the ring topology,  multicast opportunities, and computing tasks to improve the communication efficiency.  

In this paper, we consider a coded distributed computing system comprising $\sN$ nodes arranged in a ring topology, as depicted in Fig. \ref{fig: sysmodel}, where each node computes its designated output function from input files $\{w_1,\ldots, w_\sN\}$.  Key system parameters include:  1) \emph{computation load} $r\in\{1,\ldots, \sN\}$, representing that, a map function of an input file is computed averagely by $r$ nodes), 2) \emph{broadcast distance} $d\in\{1,\ldots,\lfloor\frac{\sN}{2}\rfloor\}$, specifying each node’s maximum direct connection range within the ring. We consider two typical distributed computing problems \cite{sanders2019sequential}: \emph{all-gather} where every node requires all intermediate values (IVs) mapped from input files and \emph{all-to-all} where each node requires a distinct set of IVs mapped from input files, as shown in Fig. \ref{fig: allgather and all2all}. Our goal is to design coded distributed computing to minimize the normalized communication load (NCL), defined as the communication bits normalized by the number of nodes and the bits of an IV, given any computation load $r$ and the broadcast distance $d$. 
Our main contributions are summarized as follows.
\begin{itemize}
        \item For all-gather computing, we propose a novel coded transmission scheme employing successive reverse carpooling where nodes transmit an encoded packet containing two messages traveling in opposite directions along the same path. Our scheme achieves the NCL of $\lceil \frac{\sN - r}{2d}\rceil$. We then derive a converse lower bound as $\frac{\sN - r}{2d}$, demonstrating that our scheme closely approaches the optimal NCL and exhibits asymptotic optimality when $\sN\gg d$.  
        \item For all-to-all computing, we proposed a coded transmission scheme that achieves the NCL of $O(\frac{(\sN-r)^2}{8d})$, and is asymptotically optimal under the cyclic file placement,\footnote{The cyclic placement has a simple form to be implemented, and is independent of the computation tasks, making it suitable for offline scenarios. Therefore, it has been widely adopted in many coded distributed computing, including \cite{tandon2017gradient,ye2018communication,huang2023fundamental}.} when $\sN$ is relatively large compared to $r$.  We also propose another scheme that is optimal when $d=1$ and $r\geq \sN/2$, regardless of file placement. {In addition, we establish a placement-agnostic converse bound that provides a general performance benchmark for arbitrary file assignments.}  
        \item From the optimality results of both all-gather and all-to-all computing, an interesting insight is that the computation load $r$ only leads to an additive gain in the reduction of communication load, while the broadcast distance $d$ contributes to a multiplicative gain. This behavior differs from prior results on coded distributed problems \cite{li2017fundamental,yu2017optimally,yan2017placement,cheng2024asymptotically,li2017scalable,li2019wireless,bi2024normalized,wan2020topological,hu2024exploiting}, where the computation load $r$ could lead to multiplicative gains. 

\end{itemize} 

{Many recent works in coded caching \cite{maddah2014fundamental,maddah2014decentralized}, coded distributed computing \cite{li2017fundamental,li2015coded}, and linearly separable computing \cite{khalesi2023multi} have demonstrated communication-computation gains under arbitrary, pre-set, or decentralized placements, typically assuming a shared link or fully connected network. Topology-aware coded distributed computing has further incorporated network constraints into the system design, such as fat-tree architectures \cite{wan2020topological}.  In comparison, the present work focuses on ring networks with limited broadcast distance $d$, where multicast opportunities are structurally restricted by topology, and the cyclic placement is adopted as a topology-aligned design. Notably, when $d$ is sufficiently large (\(d \geq \lfloor \sN/2 \rfloor\)), the ring network turns to a fully connected topology; for smaller $d$, the stronger topological constraints render existing shared-link schemes inapplicable without substantial modification.}

%\subsection{Organization and Notations}
The remainder of this work is organized as follows. The problem formulation and preliminary knowledge are given in Section \ref{sec: formulation}. Section \ref{sec: mainresults} provides the main results and compares the performance of the proposed scheme with the lower bound. The proofs of the information-theoretic lower bounds are provided in the appendix. Sections \ref{sec: alltoall} and \ref{sec: allunicast} present the proposed scheme for the all-gather and all-to-all problems, respectively. The work is concluded in Section \ref{sec: conclusion}.

\emph{Notation}: We use sans-serif font for constants, boldface for vectors and matrices, and calligraphic font for sets. Let $\mathbb{N}^{+}$ denote the set of positive integers. For any $k \in \mathbb{N}^{+}$, define $[k] \triangleq \{1,2,\ldots,k\}$. For any integer $a$, define $a_{\bmod \sN} \triangleq \tilde{a} \in [\sN]$ such that $(a-\tilde{a}) \bmod \sN = 0$. For a vector $[X^{(1)},\ldots,X^{(t)}]$, we write $X^t$ for brevity. The operation $\oplus$ denotes addition over a finite field.

\begin{table}[htbp]
    \centering
    \footnotesize % IEEE 标准：表格字体通常比正文小一号
    \caption{Summary of Main Notation}
    \label{tab:notation}
   { \begin{tabular}{ll}
        \toprule
        Symbol & Description \\ 
        \midrule
        $\sN$ & Number of computing nodes in the ring network \\
        $n_i$ & Node $i$, $i \in [\sN]$ \\
        $w_j$ & Input file $j$, $j \in [\sM]$ \\
        $r$ & Computation load (each file mapped by $r$ nodes on average) \\
        $d$ & Broadcast distance in the ring topology \\
        $V_i$ & Intermediate value (IV) of file $i$ needed by all nodes in all-gather problem \\
        $v_i^j$ & Intermediate value (IV) of file $j$ needed by node $i$ in all-to-all problem \\
        $\sT_1(r,d)$ & NCL for all-gather computing \\
        $\sT_2(r,d)$ & NCL for all-to-all computing \\[3pt]
        $X_i^{(t)}$ & Coded transmission from node $n_i$ at clock tick $t$ \\
        $X_i^{(j,k)}$ & Coded transmission from node $n_i$ in round $j$, step $k$ \\
        \bottomrule
    \end{tabular}}
\end{table}

\section{Problem Formulation}
\label{sec: formulation}
This section presents our problem formulation and some preliminary background knowledge in network coding.

\subsection{System Model}\label{sec: sub-sysmodel}
\begin{figure}[tbp]
        \centering
        \includegraphics[width=0.4\linewidth]{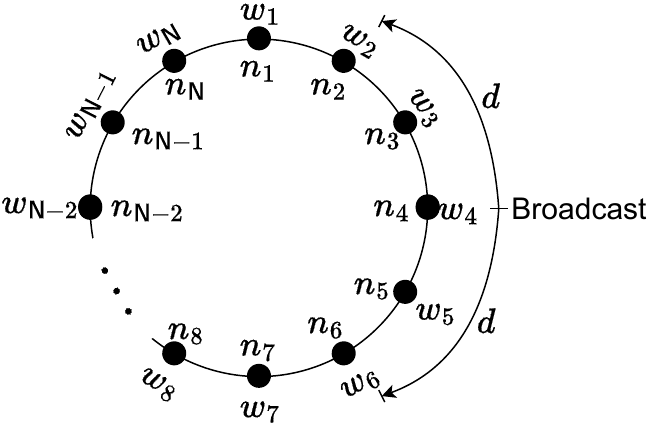}
        \caption{Ring network with $\sN$ nodes, computation load $r=1$ and broadcast distance $d=2$.}
        \label{fig: sysmodel}
\end{figure}
Consider a ring network where $\sN\in\mathbb{N}^{+}$ nodes are placed at equal distances on a circle, and each node can successfully broadcast information to some other nodes within a certain distance along the ring.\footnote{We refer to the sending method of nodes as \emph{broadcasting} following \cite{fragouli2008efficient}. However, only a subset of nodes in the network can receive the signals directly transmitted by a given node.} Specifically, the \emph{broadcast distance} $d$ means that each node $n_i$ can broadcast its message to $2d$ neighboring nodes: $\{n_{i-d}, n_{i-d+1}, \ldots, n_{i-1}, n_{i+1}, \ldots, n_{i+d-1}, n_{i+d}\}_{\bmod \sN}$. 
For example, as shown in Fig.~\ref{fig: sysmodel}, the messages broadcast by node $n_4$ can be successfully received by nodes $\{n_2,n_3,n_5,n_6\}$ when the broadcast distance is $d=2$.  Since $d\geq \lfloor\frac{\sN}{2}\rfloor$ implies full connection topology, we mainly focus on $d\in\{1,\ldots,\lfloor\frac{\sN}{2}\rfloor\}$. 

The $\sN$ nodes are assigned some computing tasks over $\sN$ input files $w_1, \ldots, w_{\sN}\in\mathbb{F}_q^{\sF}$ for some $\sF\in\mathbb{N}^{+}$, where $\mathbb{F}_q$ denotes a finite field of order $q$.\footnote{{A system with an arbitrary number of input files and output functions can be reduced to this setting via appropriate padding and batching. See Remark~\ref{convertion to mapreduce} for further discussion.}} Each node $n_k$ initially caches a set of the files $\cM_k\subseteq\{w_1,\ldots,w_\sN\}$. We assume that each input file is at least cached by one node, i.e., $\bigcup\limits_{k=1,\ldots, \sN}\cM_k= \{w_1,\ldots, w_\sN\}$. Node $n_k$ first locally computes the functions that map the files in the set $\cM_k$ into some  IVs, then exchanges IVs with neighboring nodes to compute its output function.  
We define the computation load as follows.
\begin{definition}[\emph{Computation Load}]\label{Def: computation load}
        Define the \emph{computation load}, denoted by $r$, $1\leq r\leq \sN$, as the total number of input files initially stored by the $\sN$ nodes, normalized by the number of input file $\sN$, i.e., $r=\frac{\sum_{i=1}^{\sN}|\cM_i|}{\sN}$. The computation load $r$ can be interpreted as the average number of nodes that map each input file.
\end{definition}

Suppose the network is clocked during the exchanging of IVs, i.e., a universal clock ticks $\sTe$ times.
As shown in Fig.~\ref{fig: allgather and all2all}, we consider the two computing scenarios: 1) all-gather computing and 2) all-to-all computing, where the goals in the two scenarios are decomposed as follows:

\subsubsection{All-Gather}
Each file $w_i$, $i\in[\sN]$ is first mapped into IVs
$$V_i = g_{1,i}(w_i),$$ 
where $g_{1,i}$ is a map function $\mathbb{F}_q^{\sF}\to \mathbb{F}_q^{\sB_{1}}$ for file $w_i$.
After clock tick $t-1$ and before clock tick $t$ for $t=1, \ldots, \sTe$,
node $n_k$ creates a coding symbol $X_k^{(t)} \in\mathbb{F}_{q}^{l_k^{(t)}}$, for some $l_k^{(t)}\in\mathbb{N}$, as a function of IVs computed locally and the received messages from pastime clock $Y_k^{(t-1)}, \ldots, Y_k^{(1)}$, i.e.,
% \begin{align}\label{eq: EncXi_s}
%         X_k^{(t)} = \phi_{s,k}^{(t)}\left( \{g^{1}_i: w_i\in\cM_k \}, \{Y_k^{(t-1)}, \ldots, Y_k^{(1)}\}\right),
% \end{align}
\begin{align}\label{eq: EncXi_s}
        X_k^{(t)} = \phi_{1,k}^{(t)}\left( \left( V_i: w_i \!\in\! \cM_k \right), \left( Y_k^{(t-1)}\!\!, \ldots, Y_k^{(1)} \right) \right),
\end{align}
where $\phi_{1,k}^{(t)}$ is an encoding function and $Y_k^{(i)}$ is the collection of messages received by node $n_k$ at clock $i\in[t-1]$. Having generated the messages $X_k^{(t)}$, node $n_k$ broadcasts them and nodes located within the broadcast distance, i.e., nodes in $\{n_{k-d}, n_{k-d+1}, \ldots, n_{k-1}, n_{k+1}, \ldots, n_{k+d-1}, n_{k+d}\}_{\bmod \sN}$ receive $X_k^{(t)}$ at clock tick $t$.
By the end of the clock tick $\sTe$, node $n_k\in[\sN]$ uses the received messages $\left(Y_k^{(1)}, \ldots, Y_k^{(\sTe)}\right)$ and the IVs computed locally to construct the desired IVs, i.e., 
\begin{align}\label{eq: Dec-alltoall}
        \!\!\!(V_1, \ldots, V_\sN) \!=\! \psi_{1,k} \!\left(\! \left(V_i\!:\! w_i\in\cM_k \right)\!,\!  \left(Y_k^{(1)}, \ldots, Y_k^{(\sTe)}\right) \!\right),
\end{align}  
where $\psi_{1,k}$ is an appropriate decoding function at node $n_k$. 
Finally, the node $n_k$ is responsible for computing output function $h_{1,k}$, which take all IVs as inputs, i.e., 
\begin{align*}
        D_{1,k} = h_{1,k}(V_1, \ldots, V_\sN), \qquad k\in[\sN].
\end{align*}

\subsubsection{All-to-All}
Each file $w_i$, $i\in[\sN]$ is first mapped into IVs 
$$(v_i^1, \ldots, v_i^{\sN}) = g_{2,i}(w_i),$$
where $g_{2, i}$ is a map function $\mathbb{F}_q^{\sF}\to (\mathbb{F}_q^{\sB_{2}})^{\sN}$ for file $w_i$.\footnote{We can think of $V_i$, $v_i^j$, $\forall i,j\in[\sN]$ as symbols, or as packets of symbols of the same size, with the operation applied to each packet being symbol-wise. In the following, we will use the terms 'symbol' and 'packet' interchangeably. }
The transmission process of IVs for all-to-all computing is similar to all-gather computing, while the communication messages, desired IVs, and output functions of the nodes are replaced by the following.
Node $n_k$ generates the message 
\begin{align}\label{eq: EncXi_m}
       \!\!\!\! X_k^{(t)}\! = \phi_{2,k}^{(t)}\left( \!\left(v^{j}_i\!:\! j\!\in\![\sN], w_i\!\in\!\cM_k \right), \left(Y_i^{(t\!-\!1)}, \ldots, Y_i^{(1)}\right)\!\right),
\end{align}
and desires the IVs
\begin{align}\label{eq: Dec-allunicast}
        &(v^k_1, \ldots, v^k_\sN)  = \psi_{2,k}\! \left(\left(v^{j}_i: j\!\in\![\sN], w_i\in\cM_k \right)\!, \left(Y_k^{(1)}, \ldots, Y_k^{(\sTe)}\right) \right).
\end{align} 
Finally, node $n_k$ is responsible for computing output function $h_{2,k}$, which take specific IVs as inputs, i.e., 
\begin{align*}
        D_{2,k} = h_{2,k}(v_1^{k}, \ldots, v_{\sN}^{k}), \qquad k\in[\sN].
\end{align*}

\begin{figure} [tbp]
        \subfloat[All-Gather]{
                \begin{minipage}[t]{\linewidth}
                        \centering
                        \includegraphics[width=0.55\textwidth]{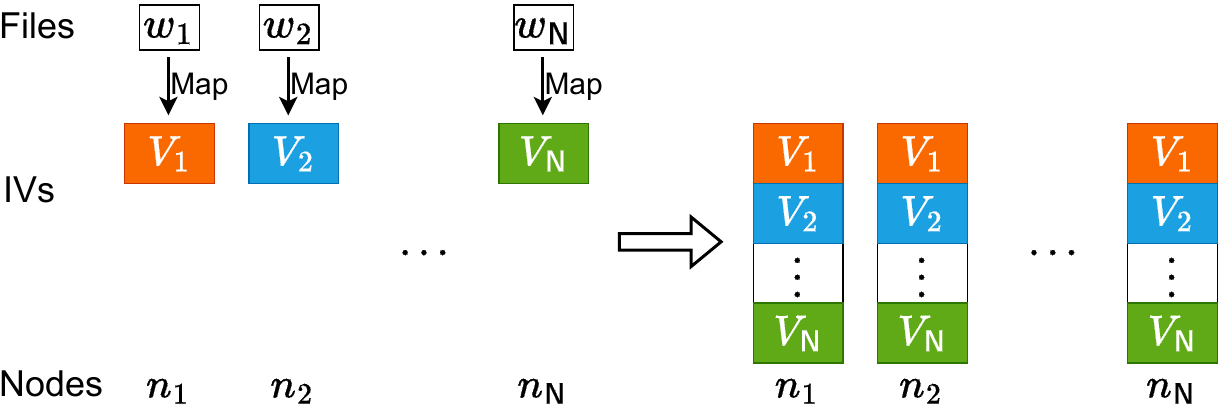}
                \end{minipage}
        }

        \subfloat[All-to-All]{
                \begin{minipage}[t]{\linewidth}
                        \centering
                        \includegraphics[width=0.55\textwidth]{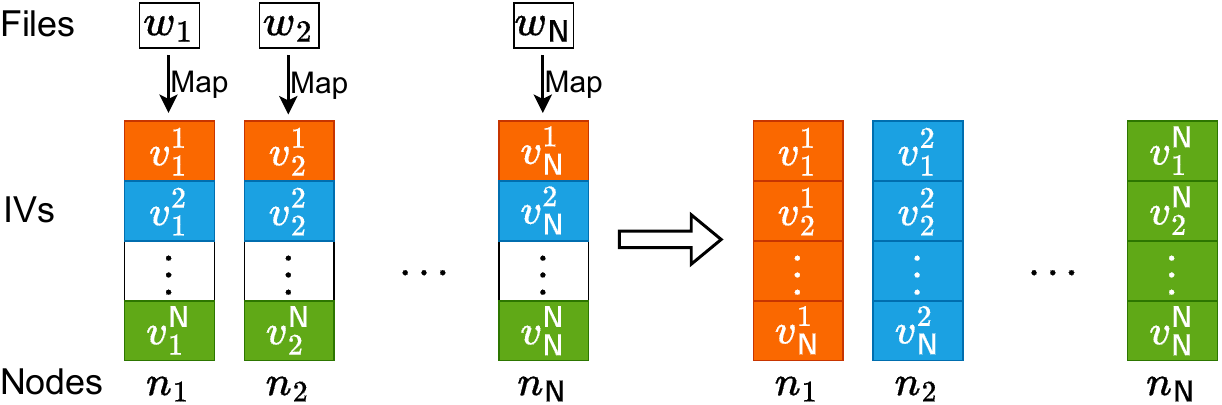}
                \end{minipage}
        }
        \caption{Two computing scenarios when $r=1$: (a) All-Gather and (b) All-to-All.\label{fig: allgather and all2all}}
\end{figure}

{Consistent with the literature on coded distributed computing, we assume that data placement occurs in a preprocessing phase. Our focus is on designing efficient transmission schemes for the delivery of IVs, which typically constitutes the communication bottlenecks. To this end}, we introduce the following metric to formulate the communication load.
\begin{definition}[\emph{Normalized Communication Load (NCL)} ]
        Given the computation load $r$ and broadcast distance $d$, we define the \emph{normalized communication loads} for {the} all-gather and all-to-all computing problems  as $\sT_{1}(r,d) = \frac{\sum_{t\in[\sTe]}\sum_{k\in[\sN]} l_k^{(t)}}{\sN \sB_{1}}$ and $\sT_{2}(r,d) = \frac{\sum_{t\in[\sTe]}\sum_{k\in[\sN]} l_k^{(t)}}{\sN \sB_{2}}$, respectively. %where $\sB=\sB_{1}$ for the all-gather problem and $\sB=\sB_{2}$ for the all-to-all problem.
        That is,  the NCL represents the normalized number of bits communicated by each node in the network. Define the minimum NCLs for the all-gather and all-to-all computing problems as $ \sT_{1}^{\star}(r,d)$ and  $\sT_{2}^{\star}(r,d)$, respectively. The objective of this paper is to characterize $\sT_1^{\star}(r,d)$ and $\sT_2^{\star}(r,d)$, given the computation load $r$ and broadcast distance $d$.
\end{definition}
Following the same definition in \cite{fragouli2008efficient}, each time a broadcast is made by a node, it is counted as one transmission, then $\sT_1$ and $\sT_2$ can be interpreted as the number of transmissions required for a node to facilitate the information exchange in the network. 

{These two problems require fundamentally new transmission and decoding designs compared with existing coded computing frameworks, due to the distance-limited broadcast constraint imposed by the ring topology. Moreover, the all-to-all problem is inherently more challenging than all-gather because information flows are directional. As will be shown in the subsequent sections, addressing these challenges requires the integration of multi-round scheduling, topology-aware coding, and successive decoding mechanisms tailored to the ring structure.}

\begin{remark}{(Latency)
     Consider a full-parallelism model where all nodes transmit simultaneously over orthogonal resources (e.g., in time, frequency, or code domains) without mutual interference. The latency can be characterized by $L(r, d) = \frac{\sum_{t\in[\sTe]}\max\left(l_1^{(t)}, \dots,  l_{\sN}^{(t)} \right)}{ \sB}$, where $\sB$ denotes per-node transmission rate (in bits per time slot). It can be observed that, in symmetric configurations where each node transmits an equal number of bits per clock tick, the latency is proportional to the NCL. Consequently, minimizing the NCL is equivalent to minimizing latency in such settings.}
\end{remark}
\begin{remark}{(Communication load of data placement)
        A worst-case analysis of the data placement phase reveals a transmission overhead of $O(\sN^2 r)$ to place $\sN r$ files, given the $O(\sN)$ maximum path length in a ring network. This results in a normalized communication load of $O(\sN r)$ per file. In many application scenarios, this cost would become negligible when the placement is amortized over multiple computation rounds. Furthermore, in scenarios where placement represents local data acquisition (e.g., sensing/satellite systems), the cost is application-dependent and remains beyond the scope of this work.}
\end{remark}
\begin{remark}\label{convertion to mapreduce}{(Relation to the MapReduce framework \cite{li2017fundamental,li2015coded}) 
        The system with an arbitrary number of files and output functions can be reduced to the above setting through appropriate padding and batching of files and functions. In particular, our model can be interpreted within the MapReduce framework~\cite{li2017fundamental} by considering the special cases where each output function is computed either by all nodes (all-gather) or by exactly one node (all-to-all), under the additional constraint that each computing node has a limited broadcast distance. Specifically, consider the MapReduce-style system with $\sN$ nodes are assigned $\sQ\in\mathbb{N}^{+}$ output functions which takes input the $\sM\in\mathbb{N}^{+}$ files. Let $g\triangleq\frac{\sM}{\sN}, p\triangleq\frac{\sQ}{\sN}\in\mathbb{N}^{+}$, which can be satisfied by injecting empty files or functions into the system when they are not integers \cite{li2017fundamental}. We can partition the $\sM$ input files into $\sN$ disjoint batches $\{\mathcal{B}_1, \ldots, \mathcal{B}_{\sN}\}$, each containing $g$ files. Let $\mathcal{W}_k\subseteq[\sQ]$ be the index set of output functions assigned to node $n_k$. 
        By viewing each batch $\mathcal{B}_i$ as a whole file $\hat{w}_i$, the MapReduce-style system can be transformed into our system model with $\sN$ files and $\sN$ nodes:
        \begin{itemize}
                \item $\mathcal{W}_i = \mathcal{W}_j$ and $|\mathcal{W}_i| = Q$ for all $i,j\in[\sN]$ corresponds to the all-gather computing, where each output function is computed by all nodes.
                \item $\mathcal{W}_i \cap \mathcal{W}_j = \emptyset$ and $|\mathcal{W}_i| = |\mathcal{W}_j| = \frac{\sQ}{\sN}$ for all $i \neq j$ corresponds to the all-to-all computing, where each output function is computed by exactly one node.
        \end{itemize}
        In the MapReduce-style system~\cite{li2017fundamental}, the computation load is defined identically to Definition~\ref{Def: computation load}. The normalized communication load is also defined analogously, differing from our definition only by an additional scaling factor. In most of the paper, we take $\sM=\sQ=\sN$; this setting simplifies exposition while preserving the generality of the provided results and analytic techniques.}
\end{remark}
\begin{remark}(Relation to the network coding problem)
       When $r=1$, the transmission of $(V_i, \ldots, V_i)$ in all-gather is also called all-to-all broadcast in network coding problem \cite{fragouli2008efficient}. If we identify the $r\geq 2$ IVs of a node as a file, and let the IVs be cyclically placed on nodes (see \eqref{eq: cyclic file}), the transmission process in all-gather can be viewed as the problem \cite{fragouli2008efficient} with correlated files of cyclic overlaps. The correlation is modeled as: the file of each $r$ neighboring nodes contains a unique common block. 
\end{remark}

\subsection{Preliminary}

\subsubsection{Reverse Carpooling}
Our transmission strategy is built on \emph{reverse carpooling}  \cite{effros2006tiling} that realizes the benefits of network coding. As demonstrated in Fig.~\ref{fig: ReverseCar}~(a), where nodes $n_1$ and $n_3$ want to exchange packets through a relay node $n_2$. After receiving packets $P_a$ and $P_b$ sent from two different nodes on the opposite sides, the relay $n_2$ broadcasts the linear combination $P_1\oplus P_2$. Each of $n_1$ and $n_3$ can then obtain the desired packet by subtracting the packet it sent previously from the mixed packet. Therefore, with network coding, it needs $3$ transmissions. In contrast, the traditional forwarding approach requires $4$ transmissions (two for each packet). In general, reverse carpooling involves two information flows that traverse a common path in opposite directions. For example, as shown in Fig.~\ref{fig: ReverseCar}(b), the path $(n_1, n_2, n_3, n_4)$ is shared by two flows, one from $n_1$ to $n_4$ and the other from $n_4$ to $n_1$. {Without reverse carpooling, each hop must forward packets for the two flows separately, resulting in double transmissions. In contrast, reverse carpooling utilizes XOR-based broadcasting for bidirectional traffic. This allows a single transmission on shared links to serve two flows simultaneously. Consequently, under a pipelined transmission schedule, the number of required transmissions on the shared path is reduced by half, leading to an approximately $50\%$ transmission saving.} There will be multiple information flows in our ring network. The reverse carpooling technique can benefit the transmission when some broadcasts form the structure as Fig.~\ref{fig: ReverseCar}(a). 

\begin{remark}
In the considered ring-based computing model, each node will serve as a source, sink, and relay node, rather than fulfilling a single role as in traditional reverse carpooling problems. 
Furthermore, unlike traditional reverse carpooling, where edge nodes can not perform packet mixture, there are no edge nodes in a ring network.
Finally, we introduce repetitive computations among nodes, and this redundancy could help improve reverse carpooling during transmissions.
\end{remark}
\begin{figure}[tbp]
        \centering
        \includegraphics[width=0.4\linewidth]{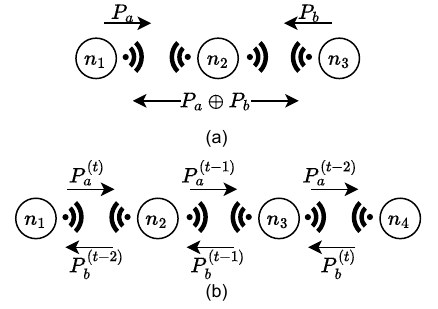}
        \caption{(a) Reverse carpooling with $3$ nodes. (b) Reverse carpooling for two flows. $P_a^{(t)}$ and $P_b^{(t)}$ are the packets sent by nodes $n_1$ and $n_4$ at time clock $t$, respectively. At clock $t$, node $n_2$ broadcasts $P_a^{(t-1)}\oplus P_b^{(t-2)}$, and node $n_3$ broadcasts $P_a^{(t-2)}\oplus P_b^{(t-1)}$. It effectively enables the two flows to traverse a common path without interfering with each other.}
        \label{fig: ReverseCar}
\end{figure}

\subsubsection{Efficient Broadcasting Using Network Coding}\label{sec: sub-tramethod}

A problem similar to transmission in all-gather computing, under the setting $r=1$ and $d=1$,
has been investigated \cite{fragouli2008efficient}. It has been shown that network coding benefits the transmission.  
The $\sN$ nodes are first partitioned in two sets $A=\{\alpha_1,\ldots,\alpha_{\frac{\sN}{2}}\}$ and $B=\{\beta_1, \ldots, \beta_{\frac{\sN}{2}}\}$ of size $\frac{\sN}{2}$ each.\footnote{$\sN$ is assumed to be even. The odd $\sN$ yields to the transmission load of the same order.} As we explain later, it is sufficient to show that each node in the sets $A$ can successfully send one information unit to all nodes in sets $A$ and $B$. Then the procedure is repeated symmetrically to send information from the nodes in sets $B$. 
The nodes communicate with other nodes in rounds, and the transmission follows two phases in each round. In the first phase, the nodes in set $A$ transmit the sum of two symbols it has received lately (or the input symbols), and the nodes in set $B$ receive. In the second phase,  the nodes in $B$ transmit and the nodes in $A$ receive. The transmission strategy is described as the algorithm \ref{alg:alg1}.
An example of the network for $\sN=8$ is depicted in Fig.~\ref{fig: TraExample n=8}, where the sets $A$ and $B$ are in color blue and red, respectively.   

\begin{figure}[tbp]
        \centering
        \includegraphics[width=0.35\linewidth]{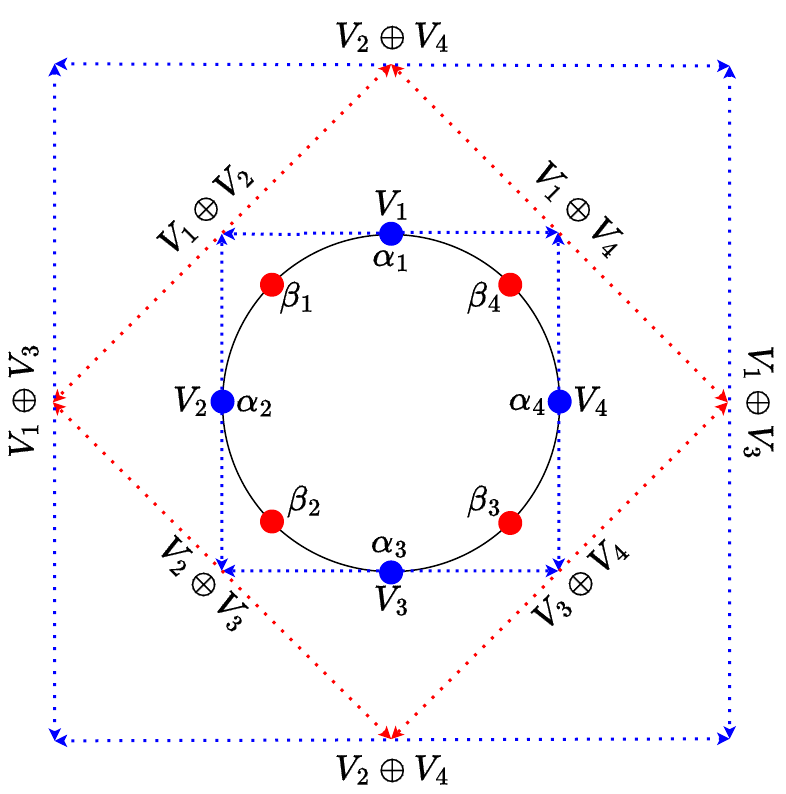}
        \caption{Efficient broadcasting over ring network with $\sN=8$ nodes, computation load $r=1$ and broadcast distance $d=1$.}
        \label{fig: TraExample n=8}
\end{figure}
\begin{algorithm}[tbp]
        \caption{Network Coding for Ring Networks \cite{fragouli2008efficient}}\label{alg:alg1}
        \begin{algorithmic}
        \STATE 
        \STATE Round $k$:
        \begin{itemize}
                \item Phase 1:
                        \\
                        if $k=1$, each $\alpha_i\in A$ transmits its information symbol $V_i$.
                        \\
                        if $k>1$, each $\alpha_i\in A$ transmits the sum of the two information symbols it received in phase $2$, round $k-1$.

                \item Phase 2:\\
                        each $\beta_i\in B$ transmits the sum of the two information symbols it received in phase $1$, round $k$.
        \end{itemize}
        \end{algorithmic}
\end{algorithm}

Following Algorithm \ref{alg:alg1},  at the end of round $k$, each node in $A$ and $B$ can receive two new symbols from the two source nodes that are $2k$ and $2k-1$ nodes away along the circle, respectively. The algorithm will finish after at most $\frac{\sN}{4}$ rounds, and all nodes can decode each input symbol. We can conclude that the NCL is 
\begin{align}\label{eq: tra_T_nc}
        \sT^{\prime}_{\text{nc}} = \frac{ (\frac{\sN}{2}+\frac{\sN}{2}) \times \frac{\sN}{4} \times 2 }{ \sN } = \frac{\sN}{2}.
\end{align} 
In the case of an odd number of nodes, a virtual node can be added to the network, allowing the partition and transmission operations of the real nodes to proceed as if the number of nodes were even. 
However, the virtual node does not transmit messages. When the messages from its neighbors are sent to the virtual node, these messages will skip the virtual node and be forwarded to the next node along the circle. With this modification, it yields the NCL of $\sT^{\prime}_{\text{nc}} = \frac{\sN-1}{2}$.
Since the broadcast distance is $d=1$, i.e., each broadcast transmission can transfer at most one symbol to two receivers, the obtained NCL is (order) optimal.

Coding in such a canonical configuration (i.e., $r=1$ and $d=1$) offers benefits. However, the problems with general $r$ and $d$, as well as for all-to-all computing, have not been sufficiently explored. In the following sections, we investigate the relationship between broadcast distance, computation load, and communication efficiency by proposing transmission strategies and providing information-theoretic analysis.

\section{Main Results}\label{sec: mainresults}

Our main results include novel achievable and converse bounds on the NCL of the coded computing system over a ring topology network. 
For the all-gather problem, the converse bound on NCL holds for any file placement.
For the all-to-all problem, the achievable bound on NCL is asymptotically optimal under the cyclic file placement. The converse bound under the cyclic file placement is provided for the computation load $r<\frac{\sN}{2}$; a converse bound for arbitrary file placement when $r\geq\frac{\sN}{2}$ and $d=1$ is also provided.

Combining the achievable and converse bounds for $r\geq 1$ and $1\leq d \leq \lfloor\frac{\sN}{2}\rfloor$, as presented in Section \ref{sec: alltoall}, the optimal NCL of the all-gather problem is as follows. 
\begin{theorem}\label{Theo: aa_uplow}
        For the all-gather computing system with a computation load $ r\in\{1,\ldots,\sN\}$ and a broadcast distance $d\in\{1,\ldots, \lfloor \frac{\sN}{2}\rfloor\}$  over the considered ring networks with $\sN$ nodes, the following  NCL is   {achievable}
        \begin{IEEEeqnarray}{rCl}\label{eq: upperThm1}
                \sT_{1}^{\tn{ach}}(r,d) \triangleq       \left\lceil \frac{\sN - r}{2d} \right\rceil.
        \end{IEEEeqnarray}
        The optimal NCL is lower bounded by
             \begin{align}\label{eq: theo1}
        \sT_{1}^{\star}(r,d) \geq \frac{\sN-r}{2d},
              \end{align} 
        {for any file placement.}
       Moreover, we have 
        \begin{IEEEeqnarray}{rCl}\label{eq:Them1Opt}
 \sT_{1}^{\tn{ach}}(r,d) -  \sT_{1}^{\star}(r,d) <1 ~\text{and}~   \lim_{{\sN}/{d}\to\infty} \frac{\  \sT_{1}^{\tn{ach}}(r,d) }{ \sT_{1}^{*}(r,d) } = 1.     
\end{IEEEeqnarray}
{Whenever $2d \mid (\sN-r)$, the proposed scheme is exactly optimal.}
\end{theorem}
\begin{proof}
        The achievable scheme is presented in Section \ref{sec: all-gather_general}.
        The proof of the lower bound is given in Appendix \ref{sec: sub-prof_sintasks_lb}. The relations in \eqref{eq:Them1Opt} are straightforward in view of \eqref{eq: upperThm1} and \eqref{eq: theo1}.
\end{proof}  
\begin{remark}\label{Re:MemShare}
   The lower convex of all points $\big\{\big(r, \left\lceil \frac{\sN - r}{2d} \right\rceil \big): r\in\{1,\ldots,\sN\}\big\}$ is achievable for general  $1\leq r\leq \sN$ by using memory-sharing.
\end{remark}

\begin{remark}\label{rem: code gain}
    Similar to \cite{li2017fundamental, huang2024coded}, the factor $(\sN-r)$ can be referred to as the local computation gain, which is a common additive gain arising from redundant computation in many coded computing schemes in the literature. 
    In addition, we observe a multiplicative factor of $2d$, which we refer to as the coded transmission gain. This gain is also observed in other works \cite{li2017fundamental, huang2024coded}, typically associated with redundant computation. In those works, IVs can be used to generate multicast messages desired by multiple nodes due to redundant computation. For example, the coded transmission gain in \cite{li2017fundamental} is equal to the value of the computation load, which implies that it vanishes when the computation load $r=1$. However, in our system, the coded transmission gain arises from a topology-based coding method, which is not constrained by the computation load and instead benefits from the connectivity of the nodes.
\end{remark}

\newcommand{\ceNr}{\left\lceil(\sN - r)/2\right\rceil}
\newcommand{\flxd}[1]{\left\lfloor\frac{#1}{d}\right\rfloor}
For the all-to-all computing over the ring-based networks, we obtain the achievable NCL in Theorem \ref{Theo: allunicast} and a lower bound on NCL under the cyclic placement in Theorem \ref{Theo: au_con-cyc}.
\begin{theorem}\label{Theo: allunicast}
      For the all-to-all computing system with a given computation load $r\in\{1,\ldots,\sN\}$ and a broadcast distance $d\in\{1,\ldots, \lfloor \frac{\sN}{2} \rfloor\}$ over the considered ring networks with $\sN$ nodes, the achievable NCL is given as follows:
        {\begin{align}
                &\sT_{\tn{2-cyc}}^{\tn{ach}}(r, d) = \left\{ \hspace{-5pt} \begin{array}{cc}\vspace{5pt} 
                        \frac{d}{2}\flxd{\ceNr}^2 + \frac{d}{2}\flxd{\ceNr} + \flxd{\ceNr}\left(\ceNr\right)_d + \ceNr, & \!2r-1 \leq \!d\!, \\\vspace{5pt} 
                        \frac{d}{2}\!\! \flxd{\ceNr - r + 1}^2 \!\!+\! \frac{3d}{2}\flxd{\ceNr - r + 1} \!\!+\! r \!-\! 1 \!+\! \left(\flxd{\ceNr - r + 1} \!\!+\!\! 2\right)\left(\lceil\frac{\sN -r}{2}\rceil \!-\! d\right)_d, & \!\!r\!-\!1 \!< \!d\! \leq \!2(r\!-\!1), \\
                        \frac{d}{2}\!\flxd{\ceNr}^2 \!+\! \frac{d}{2}\! \flxd{\ceNr} \!+\! \left(\flxd{\ceNr} \!+\! 1\right) \left(\ceNr \!-\! d\right)_d,  &  \!d\! \leq r\!-\!1,
                \end{array} \right.\nonumber
        \end{align}}
{where $(x)_d$ denotes $x \bmod d$ for $x \in \mathbb{Z}$, and the dominant term of $\sT_{\tn{2-cyc}}^{\tn{ach}}(r, d)$ is $\frac{\sN^2}{8d}$ for different ranges of $d$ and $r$.}
\end{theorem}
\begin{proof}
      Please see the scheme presented in Section \ref{sec: all-to-alls_general}.
\end{proof}

\begin{theorem}\label{Theo: au_con-cyc}
        {For the all-to-all computing system with a given computation load $r\in\{1,\ldots,\sN -1\}$ and a broadcast distance $d\in\{1,\ldots, \lfloor \frac{\sN}{2} \rfloor\}$ over the considered ring networks with $\sN$ nodes, a lower bound for the NCL  under the cyclic placement is given as:
        \begin{align}
                \sT_{\tn{2-cyc}}^{\star}(r,d) \geq \max_{s\in\{1, \ldots \sN\}} \frac{s(\sN - s - r + 1)}{2d},
        \end{align}
        where the right-hand side is maximized at $s = \lceil\frac{\sN-r+1}{2}\rceil$, giving the scaling of lower bound as $\frac{(\sN - r + 1)^2}{8d}$.}
\end{theorem}
\begin{proof}
        Please see the converse proof in Appendix \ref{sec: all-to-alls_lb_placement}.
\end{proof}

{To expose the asymptotic scaling more transparently (ignoring integer rounding), the achievable NCL in Theorem \ref{Theo: allunicast} can also be expressed as 
\begin{align}
        &\sT_{\tn{2-cyc}}^{\tn{ach}}(r, d) = \left\{ \begin{array}{cc}\vspace{5pt} 
                O\left( \frac{\sN}{4d}(\frac{\sN}{2}\!\!-\!\!r) \!+\! \frac{3(\sN-r)}{4} \!+\! \frac{r^2}{8d} \right), & \!2r-1 \leq \!d,\! \\\vspace{5pt} 
                \!\!\!O\!\!\left( \frac{\sN}{4d}(\frac{\sN}{2}\!\!-\!\!3r) \!+\! \frac{3\sN-5r}{4} \!+\! \frac{9r^2+4(\sN-r+1)}{8d} \right), & \!\!r\!-\!1 \!< \!d\! \leq \!2(r\!-\!1), \\%\substack{ r-1 \!< d \leq 2(r\!-\!1)}\\
                O\left(  \frac{\sN}{4d}(\frac{\sN}{2}\!\!-\!\!r) \!+\! \frac{(\sN-r)}{4}\!+\! \frac{r^2}{8d} \right),  &  \!d\! \leq r\!-\!1. %\substack{d\leq r-1}
        \end{array} \right.\nonumber
\end{align} 
It is readily observed that the achievable NCL in Theorem \ref{Theo: allunicast} are order-optimal. Specifically, for $\sN\gg r$, the multiplicative gap is bounded as $\frac{\sT_{\tn{2-cyc}}^{\tn{ach}}(r,d)}{\sT_{\tn{2-cyc}}^{\star}(r,d)} \leq \frac{\sN^2}{(\sN - r + 1)^2} + \epsilon $, where $\epsilon \to 0$ as $\sN \to \infty$. This ratio asymptotically approaches $1$ as $\sN \to \infty$.}
Next, we further present a converse bound for arbitrary file placement when $r\leq \frac{\sN}{2}$ and $d\in\{1,\ldots, \lfloor \frac{\sN}{2} \rfloor\}$.
\begin{theorem}\label{Theo: au_con-arb}
        {For the all-to-all computing system with a given computation load $r\in\{1,\ldots,\lfloor\frac{\sN}{2}\rfloor\}$ and a broadcast distance $d\in\{1,\ldots, \lfloor \frac{\sN}{2} \rfloor\}$ over the considered ring networks with $\sN$ nodes, a lower bound for the NCL is given as:
        \begin{align}
                \sT_{\tn{2}}^{\star}(r,d) \geq \max_{s\in\{1, \ldots \sN\}} \frac{s(\sN - sr)}{2d},
        \end{align}
        where the right-hand side is maximized at $s = \lceil\frac{\sN}{2r}\rceil$ (or $\lfloor\frac{\sN}{2r}\rfloor$), giving the scaling of lower bound as $\frac{\sN^2}{8dr}$.}
\end{theorem}
\begin{proof}
        Please see the converse proof in Appendix \ref{sec: all-to-alls_lb_placement}.
\end{proof}
{The performance gap between the achievable NCL in Theorem \ref{Theo: allunicast} and the converse bound for arbitrary placement in Theorem \ref{Theo: au_con-arb} is characterized $\frac{\sT_{\tn{2-cyc}}^{\tn{ach}}(r,d)}{\sT_{\tn{2}}^{\star}(r,d)} \leq \frac{(\sN-r)^2}{\sN^2 / r} + \epsilon$, where $\epsilon \to 0$ as $\sN \to \infty$. For small values of $r$, this ratio asymptotically approaches $r$.  This gap is expected since Theorem \ref{Theo: au_con-arb} provides a general lower bound over all placements, while Theorem \ref{Theo: allunicast} is obtained under the cyclic placement.
Our results indicate that the cyclic placement is particularly competitive when the computation load is relatively small, which is reasonable in practical distributed computing applications. A figure illustrating this comparison is provided in Fig.~\ref{fig: all2all_lb_placement}.}
\begin{figure}[tbp]
        \centering
        \includegraphics[width=0.5\linewidth]{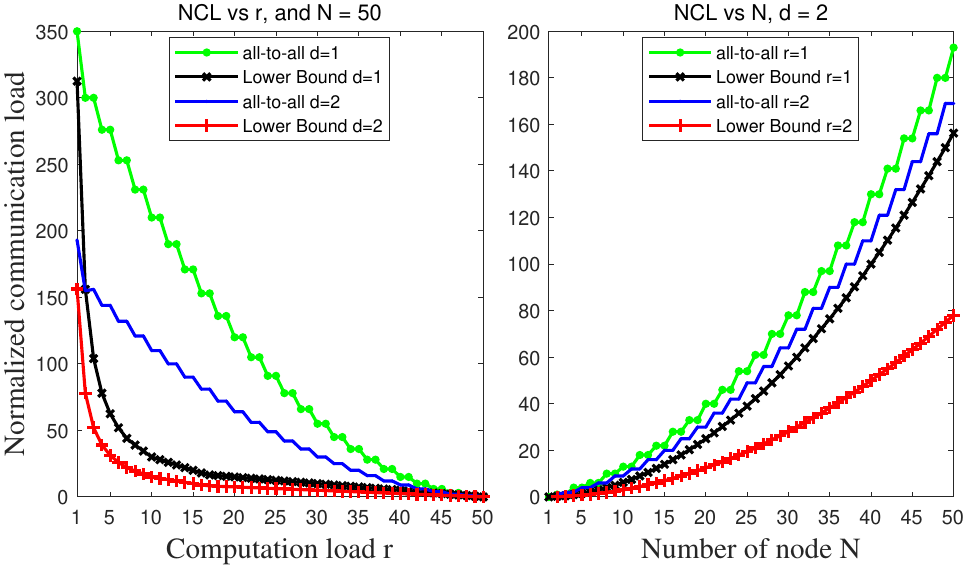}
        \caption{{Comparison of achievable NCL and lower bound of all-to-all. The achievable NCL is obtained under the cyclic placement, while the lower bound is derived for arbitrary placement.}}
        \label{fig: all2all_lb_placement}
\end{figure}
\begin{figure}[tbp]
\begin{minipage}[t]{0.48\linewidth}
        \centering
        \includegraphics[width=\linewidth]{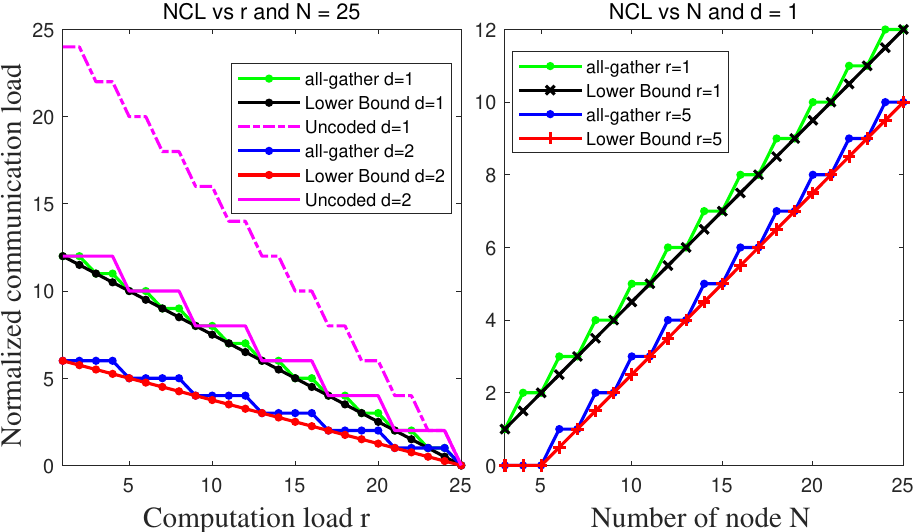}
        \caption{Comparison of achievable NCL and lower bound of all-gather.}
        \label{fig: uplow_allgather}
\end{minipage}\hspace{9pt}
\begin{minipage}[t]{0.48\linewidth}
        \centering
        \includegraphics[width = \linewidth]{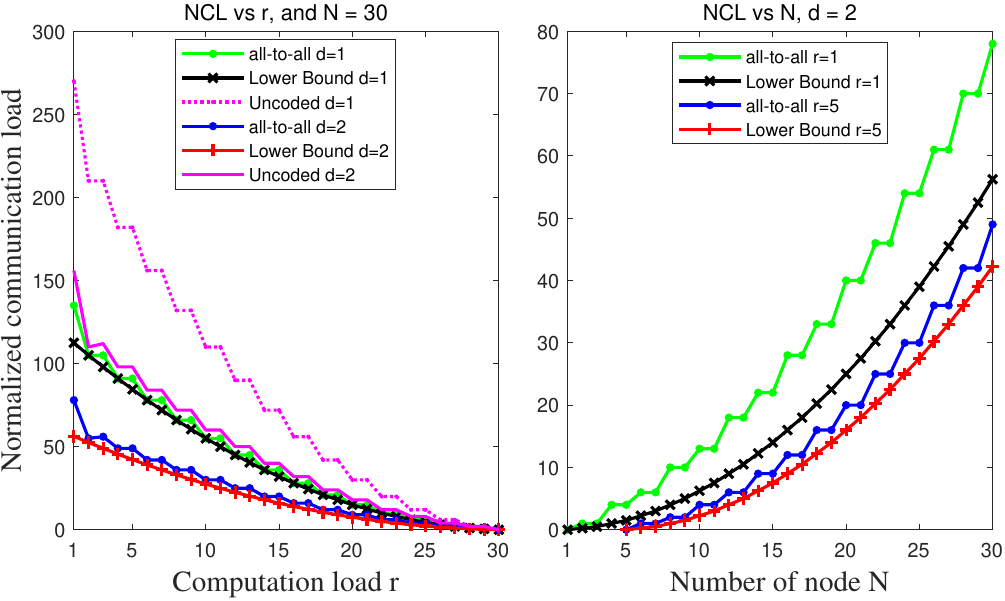}
        \caption{{Comparison of achievable NCL and lower bound of all-to-all under cyclic file placement.}}
        \label{fig: uplow_all2all}
\end{minipage}
\end{figure}

When the computation load is relatively large  $r\in\{ \lceil\frac{\sN}{2}\rceil,\ldots,\sN\}$, without constraints on data placement, an optimal result is obtained when $d=1$ as shown in the following theorem. 
 \begin{theorem}\label{Lem: au_con-arb}
        For the all-to-all computing system with a given computation load $r\in\{\lceil\frac{\sN}{2}\rceil,\ldots,\sN\}$
        and a broadcast distance $d=1$ over the considered ring network with $\sN$ nodes, the optimal NCL is given as
        \begin{align}
                \sT^{\star}_{2}(r,1)  = \frac{\sN-r}{2}. %\leq \left\lceil \frac{\sN-r}{2} \right\rceil.
         \end{align} 
\end{theorem}
\begin{proof}
  To achieve the optimal NCL,  we assign files so that the IVs are near the desired nodes, rather than adhering to the cyclic placement. Please see the proposed scheme and converse proof in  Appendix \ref{sec: allunicast-arb}.
\end{proof}

%This is expected. If the file placement can be arbitrarily desired, we can find a strategy that significantly improves communication efficiency for larger values of $r$ by placing IVs near the desired nodes.
\begin{remark}
        A lower bound when $d \geq 2$ and $r\in\{\lceil\frac{\sN}{2}\rceil,\ldots,\sN\}$ can be also derived as $\frac{\sN - r}{2d} \leq \sT^{\star}_{2}(r,d)$ using a proof similar to that in Theorem \ref{Lem: au_con-arb}. However, it is an open problem to check whether it is tight. 
        %The corresponding upper bound is left for future work.
\end{remark}

Similar to the observations in Remark \ref{rem: code gain}, NCL demonstrates that the redundant computation provides an additive gain, while the connectivity of the nodes contributes a multiplicative gain.
\begin{remark}(Uncoded Schemes)\label{rem: Uncoded}
   Follow the same steps in the coded schemes in Sections \ref{sec: alltoall} and \ref{sec: allunicast}, but letting each node directly forward each intermediate value instead of generating coded messages, the NCL of the uncoded schemes for all-gather and all-to-all are $\sT^{\tn{uncoded}}_{1}(r,d) = 2\sT_{1}(r,d)$ and $\sT^{\tn{uncoded}}_{\tn{2-cyc}}(r,d) = 2\sT^{\tn{ach}}_{\tn{2-cyc}}(r,d),$ respectively. {This factor-of-two gap arises because local broadcasts without coding cannot simultaneously satisfy distinct demands flowing in opposite directions.}
   {This indicates that ring topology would constrain the multicast opportunities, i.e., the proposed coded scheme can at most offer a multiplicative factor of $2$ of benefit compared with the uncoded scheme. Notably, prior work in network coding has shown that coding gain is bounded by a factor of $2$ for the information transmission in various undirected network model \cite{li2009constant,kiraly2008approximate}, while they do not consider the redundant computation (which can be modeled by allowing the same source messages to be sent from multiple nodes) and assume the network only has a single communication session. }
\end{remark}

Fig.~\ref{fig: uplow_allgather} compares the NCLs of the proposed scheme, the uncoded scheme under the cyclic placement, and the lower bound {(under any data placement)} for the all-gather problem when $r$ or $\sN$ changes. It can be seen that our scheme can significantly reduce the communication load compared to the uncoded scheme. The zigzag shape of the achievable NCL is due to the operation of the ceil in \eqref{eq: upperThm1}.

The left panel of Fig.~\ref{fig: uplow_all2all} compares the NCLs versu $r$ for the all-to-all problem with $\sN=50$ and $d=1,2$, including the NCLs of the proposed scheme in Theorem \ref{Theo: allunicast}, the uncoded scheme in Remark \ref{rem: Uncoded}, and the lower bounds in Theorem \ref{Theo: au_con-cyc} and Theorem \ref{Lem: au_con-arb}.  It can be seen that our scheme can significantly reduce the communication load compared to the uncoded scheme. 
The right panel of Fig.~\ref{fig: uplow_all2all} compares the NCL of the proposed scheme and the lower bound as $\sN$ varies for $r=1$ and $r=5$.
It can be seen that the proposed scheme is asymptotically optimal when $\sN$ is relatively large compared to the computation load $r$.
Additionally, the NCL approaches the lower bound when $r$ is close to $\sN$. 
For both all-to-all and all-gather, the NCLs of the proposed scheme decrease with $r$ and $d$. 

\begin{remark}({Advantages} of Parallelism)
    An advantage of our proposed schemes is their potential for parallelism. For example, when $d=1$, the nodes in $\{n_{i}, n_{i+2}, n_{i+4}, \ldots \}$ can transmit messages concurrently during a given time step without interfering with each other, and the nodes in $\{n_{i+1}, n_{i+3}, n_{i+5}, \ldots \}$ receive their intended messages. In the next step, the roles of the sending and receiving nodes are reversed to send other messages. This enables a more efficient communication bandwidth utilization, compared to the work in~\cite{li2017fundamental} where nodes contend for access to a shared communication link.
\end{remark}

\section{Coded Transmission for All-gather Computing}
\label{sec: alltoall}
This section focuses on the all-gather computing system where each node $n_i,i\in[\sN]$ requires all IVs. 
To enhance communication efficiency, the IVs are propagated to all nodes through reverse carpooling, after which locally generated IVs are used to decode the desired ones. In certain special cases, successive decoding is employed during the decoding process.
We first present an illustrative example, followed by a description of our general scheme. 
For brevity, unless otherwise specified, the notation of operation \{$\bmod\ \sN$\} for the indices of nodes and files will be omitted from here on.

\subsection{An Illustrative Example}
\label{sec: allgather-example}
Consider a ring network with $\sN = 8$ nodes, computation load $r = 2$, and broadcast distance $d=3$. 
The nodes are responsible for computing tasks on the $\sN=8$ given input files. Each node locally maps $2$ files, computing one IV from each mapped file. 
The input file is initially cached at the nodes based on a cyclic placement, where file $w_i$ is stored by nodes $\{n_i, n_{i-1}\}$. This corresponds to the placement of IVs shown in Fig.~\ref{fig: NewExample n=8, r=2, d=3}. 
Without loss of generality, we can treat each file $V_i$, $i\in[\sN]$ as a single symbol.
%Equivalently, we have the placement of IVs as shown in Fig. 7

\begin{figure}[t]
        \centering
        \includegraphics[width=0.3\linewidth]{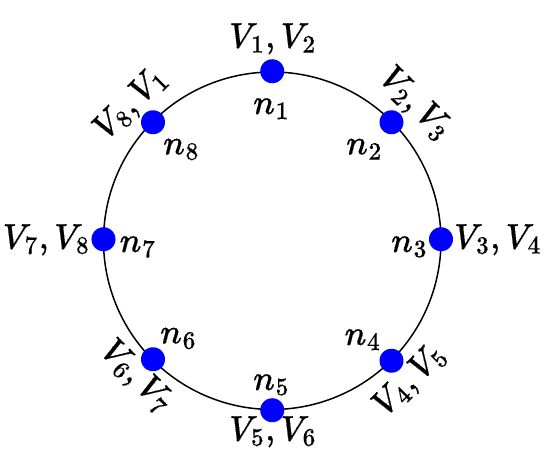}
        \caption{The IVs placement over a ring network with $\sN=8$ nodes, computation load $r=2$ and broadcast distance $d=3$.}
        \label{fig: NewExample n=8, r=2, d=3}
\end{figure}

In this example, a single broadcast per node is sufficient for each node to obtain all the desired IVs. For each node $n_i, i\in[8]$, the encoded symbol $V_i \oplus V_{i+1}$ is broadcast to the nodes $\{n_{i-3}, n_{i-2}, n_{i-1}, n_{i+1}, n_{i+2}, n_{i+3}\}$. Obviously, the NCL is $\sT_{1} = 1$.
Meanwhile, node $n_i$ can receive $6$ encoded symbols from these nodes, as follows: 
\begin{align}
        &V_{i-3} \oplus V_{i-2},& & V_{i-2} \oplus V_{i-1},&&V_{i-1} \oplus V_{i},& \nonumber\\
        & V_{i+1} \oplus V_{i+2},&&V_{i+2} \oplus V_{i+3}, && V_{i+3} \oplus V_{i+4}.& \nonumber
\end{align}
For example, node $1$ broadcasts the encoded symbol $V_1 \oplus V_2$ to the nodes $\{n_6, n_7, n_8, n_2, n_3, n_4\}$, and receives the following encoded symbols:
\begin{align}
        &V_{6} \oplus V_{7},& & V_{7} \oplus V_{8},& &V_{8} \oplus V_{1},&\nonumber\\
        &V_{2} \oplus V_{3},& & V_{3} \oplus V_{4},& &V_{4} \oplus V_{5}.& \nonumber
\end{align}
Obviously, node $n_i$ can decode $V_{i-1}$ and $V_{i+2}$ from $V_{i-1} \oplus V_{i}$ and $V_{i+1} \oplus V_{i+2}$, as it already caches $V_i$ and $V_{i+1}$. Then, using the decoded symbols $V_{i-1}$ and $V_{i+2}$, node $n_i$ can then proceed to decode new symbols $V_{i-2}$ and $V_{i+3}$. This successive decoding process continues further for $V_{i-3}$ and $V_{i+4}$.
In other words, each node decodes symbols sequentially, starting from those sent by nearby nodes and progressing to those from farther nodes.
The decoding order is presented in Table~\ref{tab:decoding oreder example}.  The symbol `$*$' at  row $n_i$ and column $V_j$, where $i\in[8]$ and $j\in[8]$, indicates that  node $n_i$ generates (and thus has) the IV $V_j$. In row $n_i$,  integer $1$ indicates that the IV in the corresponding column can be directly decoded by node $n_i$ upon receiving the encoded symbol broadcast by its neighboring nodes. The integer $2$ signifies that decoding the corresponding IV depends on the symbol associated with the nearest integer $1$ in the same row. Similarly, the integer $3$ denotes that decoding the corresponding IV relies on the symbol associated with the nearest integer $2$ in the same row.
\begin{table}[tbp]
        \caption{Decoding order for $\sN=8$, $r=2$ and $d=3$ \label{tab:decoding oreder example}}
        \centering
        \begin{tabular}{|c|c|c|c|c|c|c|c|c|}
        \hline
          & $V_1$ & $V_2$ & $V_3$ & $V_4$ & $V_5$ & $V_6$ & $V_7$ & $V_8$ \\
        \hline
        $n_1$ & $*$ & $*$ & \bl{1} & \bl{2} & \bl{3} & \bl{3} & \bl{2} & \bl{1} \\
        \hline
        $n_2$ & \bl{1} & $*$ & $*$ & \bl{1} & \bl{2} & \bl{3} & \bl{3} & \bl{2} \\
        \hline
        $n_3$ & \bl{2} & \bl{1} & $*$ & $*$  & \bl{1} & \bl{2} & \bl{3} & \bl{3} \\
        \hline
        $n_4$ & \bl{3} & \bl{2} & \bl{1} & $*$ & $*$  & \bl{1} & \bl{2} & \bl{3} \\
        \hline
        $n_5$ & \bl{3} & \bl{3} & \bl{2} & \bl{1} & $*$ & $*$  & \bl{1} & \bl{2} \\
        \hline
        $n_6$ & \bl{2} & \bl{3} & \bl{3} & \bl{2} & \bl{1} & $*$ & $*$  & \bl{1} \\
        \hline
        $n_7$ & \bl{1} & \bl{2} & \bl{3} & \bl{3} & \bl{2} & \bl{1} & $*$ & $*$  \\
        \hline
        $n_8$ & $*$  & \bl{1} & \bl{2} & \bl{3} & \bl{3} & \bl{2} & \bl{1} &  $*$ \\
        \hline
        \end{tabular}
\end{table}

\subsection{General Scheme}
\label{sec: all-gather_general}
Consider the integer-valued computation load $r\in\{1,\ldots, \sN\}$ and broadcast distance $d\in\{1,\ldots, \lfloor\frac{\sN}{2}\rfloor \}$.
Without loss of generality, we treat the IVs mapped from input files as individual symbols. When $r=\sN$, each node knows all the desired IVs, thus no communication is needed, resulting in $\sTnc(r=\sN, d) = 0$ for all $d\in\{1,\ldots, \lfloor\frac{\sN}{2}\rfloor \}$. Next, we focus on the case where $r< \sN$. 

% Transmission at round $1$.
The input files are initially cached by the nodes based on a cyclic placement along the ring network. 
Specifically, node $n_i$ caches input files $w_j$ if $i \leq j \leq i+r-1$, i.e.,
\begin{align}\label{eq: cyclic file}
        \cM_i = \{w_i, \ldots, w_{i+r-1}\}.
\end{align}
After performing the local mapping, node $n_i$ knows the IVs corresponding to the files in $\cM_i$, i.e., $\{V_n, w_n\in\cM_i\}$.

At time clock $1$, node $n_i$ generates the encoded symbol 
\begin{align*}
        X_{i}^{(1)} = V_i \oplus V_{i+r-1},
\end{align*}
and broadcasts it to the nodes within a distance of $d$; thus the set of nodes $\{n_{i-d}, \ldots, n_{i-1}, n_{i+1}, \ldots, n_{i+d} \}$ can receive $X_{i}^{(1)}$. Meanwhile, node $n_i$ can receive the encoded symbols sent by these nodes, specifically
% $Y_{i_{-}}^{(1)}=( X_{i-d}^{(1)},\ldots, X_{i-1}^{(1)})$ and $Y_{i_{+}}^{(1)}=( X_{i+1}^{(1)},\ldots,X_{i+d}^{(1)})$, 
\begin{align*}
        Y_{i_{-}}^{(1)}  &=  \left(X_{i-d}^{(1)} \!=\! V_{i-d} \oplus\! V_{i-d+r-1}, \ldots, X_{i-1}^{(1)} \!=\! V_{i-1} \oplus\! V_{i+r-2} \right)\nonumber
\end{align*}
\begin{align*}
        Y_{i_{+}}^{(1)} & = \left( X_{i+1}^{(1)} \!=\! V_{i+1} \oplus\! V_{i+r}, \ldots, X_{i+d}^{(1)} \!=\! V_{i+d} \oplus\! V_{i+d+r-1} \right),\nonumber
\end{align*}
where $Y_{i_{-}}^{(1)}$  and $Y_{i_{+}}^{(1)}$ denote the collection of the messages transmitted by neighboring nodes $\{n_{i-d}, \ldots, n_{i-1}\}$ and $\{n_{i+1}, \ldots, n_{i+d}\}$ at clock $1$, respectively.

We now present the decoding process at time clock $1$. If $ d \leq r-1 $, we observe that the components $ \left(V_{i-d+r-1}, \ldots, V_{i+r-2}\right) $ from $ Y_{i_{-}}^{(1)} $ and the components $ \left(V_{i+1}, \ldots, V_{i+d}\right) $ from $ Y_{i_{+}}^{(1)} $ are already known by node $ n_i $.
Thus, node $n_i$ can successfully decode $2d$ desired symbols as follows,
\begin{align*}
        (V_{i-d}, \ldots, V_{i-1}, V_{i+r}, V_{i+r+d-1}).
\end{align*}
%$$v_{i-d}, \ldots, v_{i-1}, v_{i+r}, v_{i+r+d-1}.$$ 
If $d\geq r$, node $n_i$ can also decode $2d$ desired symbols but with successive decoding.
Taking $Y_{i_{-}}^{(1)}$ as an example (illustrated in Fig.~\ref{fig: Round1-Successive Decoding}), the same color on the two sides of the summations means the different coded symbols share components. Since node $n_i$ has already known the right-hand side of the first $r-1$ summations (i.e., $V_{i+r-2},\ldots, V_{i}$), it can decode their left-hand side. These decoded symbols then serve as known components for the next $r-1$ summations (i.e., the $r^{\textnormal{th}}$ to $(2r-1)^{\textnormal{th}}$ summations from the bottom), allowing further decoding in a successive manner. Following this process, node $n_i$ successively decodes $2d$ desired symbols 
$$\left(V_{i-1}, \ldots, V_{i-d}\right) \text{ and } \left(V_{i+r}, \ldots, V_{i+d+r-1}\right)$$
from $Y_{i_{-}}^{(1)}$ and $Y_{i_{+}}^{(1)}$, respectively.
\begin{figure}[tbp]
        \centering
        \includegraphics[width=0.2\linewidth]{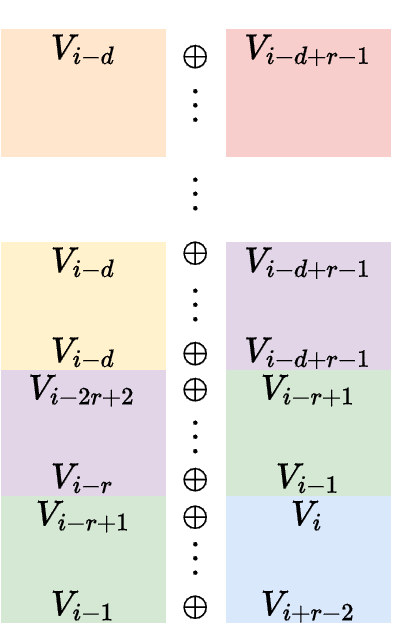}
        \caption{Successive decoding of $Y_{i_{-}}^{(1)}$ at time $1$.  The same color on the two sides of the summations means the same component, but constructing different encoded symbols.}
        \label{fig: Round1-Successive Decoding}
\end{figure}

We now describe the transmission process from time clock $k\geq 2$. After time $k-1$, node $n_i$ has successfully decoded and stored the following symbols, 
$$V_{i-d(k-1)}, \ldots, V_{i-1}, V_{i}, V_{i+1}, \ldots, V_{i+d(k-1)+r-1}. $$ 
At time clock $k$, the symbols coming from the opposite sides, $V_{i-d(k-1)}$ and $V_{i+d(k-1)+r-1}$, are mixed to generate the coded symbols as
% The encoded symbols generated by node $n_i$ for transmission at time clock $k$ is 
\begin{align*}
        X_{i}^{(k)} = V_{i-d(k-1)} \oplus V_{i+d(k-1)+r-1}.
\end{align*}
$X_{i}^{(k)}$ is then broadcast to nodes $\{n_{i-d}, \ldots, n_{i-1}, n_{i+1}, \ldots, n_{i+d} \}$, where half of these nodes have already known $V_{i-d(k-1)}$ and the other half have known $V_{i+d(k-1)+r-1}$, forming the reverse carpooling structure. 
Meanwhile, node $n_i$ can receive the coded symbols sent from these nodes, specifically
% specifically $Y_{i_{-}}^{(k)} = (X_{i-d}^{(k)}, \ldots, X_{i-1}^{(k)})$ and $Y_{i_{+}}^{(k)} = (X_{i+1}^{(k)}, \ldots, X_{i+d}^{(k)})$, 
\begin{align*}
    & Y_{i_{-}}^{(k)} = \left( X_{i-d}^{(k)} \!=\! V_{i-dk} \oplus\! V_{i+r-1+d(k-2)}, \ldots,  X_{i-1}^{(k)} \!=\! V_{i-1-d(k-1)} \oplus\! V_{i+r-2+d(k-1)} \right)
\end{align*}
and 
\begin{align*}
    & Y_{i_{+}}^{(k)} = \left( X_{i+1}^{(k)} \!=\! V_{i+1-d(k-1)} \oplus\! V_{i+r+d(k-1)}, \ldots,  X_{i+d}^{(k)} \!=\! V_{i - d(k-2)} \oplus\! V_{i+r-1+dk} \right),
\end{align*}
where $Y_{i_{-}}^{(k)}$  and $Y_{i_{+}}^{(k)}$ denote the collection of the messages transmitted by neighboring nodes $\{n_{i-d}, \ldots, n_{i-1}\}$ and $\{n_{i+1}, \ldots, n_{i+d}\}$ at clock $k$, respectively.
The successive decoding is no longer required here. In $Y_{i_{-}}^{(k)}$, the right-hand side of each summation is the content already stored by node $n_i$, and the left-hand side of each summation is the desired symbols. Hence, node $n_i$ can successfully decode the following symbols 
$$ V_{i-dk}, \ldots,  V_{i-1-d(k-1)}.$$
Symmetrically, in $Y_{i_{+}}^{(k)}$, the left-hand side is the stored symbols, while the right-hand side is the desired ones, allowing node $n_i$ to decode 
$$ V_{i+r+d(k-1)}, \ldots,  V_{i+r-1+dk}.$$
At the end of time clock $k$, node $n_i$ obtains $2d$ new symbols. The symbols that are decoded and stored by node $n_i$ are
$$V_{i-dk}, \ldots, V_{i-1}, V_{i}, V_{i+1}, \ldots, V_{i+dk+r-1}. $$ 
The process continues in subsequent times until each node receives all the desired symbols.

Since each node can obtain at most $2d$ new symbols through one transmission, the algorithm completes in time $\sTe = \lceil \frac{\sN - r}{2d} \rceil $. 
The ceiling operation accounts for the final round, where the number of remaining symbols for each node may be less than $2d$, yet every node still performs one transmission.
Thus, the NCL is given by 
\begin{align}\label{eq: general NCL}
        \sT_{1}(r, d) = \frac{ \lceil \frac{\sN - r}{2d} \rceil \sN \sB_{1}}{\sN \sB_{1}} = \left\lceil \frac{\sN - r}{2d} \right\rceil.
\end{align}

\section{Coded Transmission for All-to-all Computing}\label{sec: allunicast}
% \section{ALL-Unicast with Caches}\label{sec: allunicast}
This section focuses on the all-to-all problem where each node $n_i, i\in[\sN]$, requires a distinct set of IVs. Instead of repeating our all-gather scheme multiple rounds, we delicately deliver the IVs according to their distance to the intended nodes. In this organized manner, the IVs efficiently reach the intended node.
We first present an illustrative example, followed by a description of our general scheme. 

For clarity, we redefine the notation for transmitted messages. The message $X_{i}^{(t)}$, transmitted by node $n_i$ at time slot $t$, is now written as $X_{i}^{(j,k)}$. Here, $X_{i}^{(j,k)}$ represents the message transmitted by node $n_i$ at the $k^{\textnormal{th}}$ step of round $j$, where $0 \leq k \leq j$ and $t = 1 + 2 + \ldots + (j-1) + k$. Similarly, we re-write $Y_i^{(t)}$ as $Y_i^{(j,k)}$, maintaining the same relationship $0\leq k\leq j$ and $1+ 2 + \ldots + (j-1) + k = t$.

\subsection{An Illustrative Example}\label{sec: MUexa}

Consider a ring network with $\sN = 8$ nodes, computation load $r = 3$, and broadcast distance $d=1$.  As shown in Fig.~\ref{fig: MultiUnicast n=8, r=3, d=1}, the input files are cached by nodes based on a cyclic placement, i.e., files $w_i$ are stored by nodes $\{n_i, n_{i-1}, n_{i-2}\}$. 

\begin{figure}[tbp]
        \centering
        \includegraphics[width=0.35\linewidth]{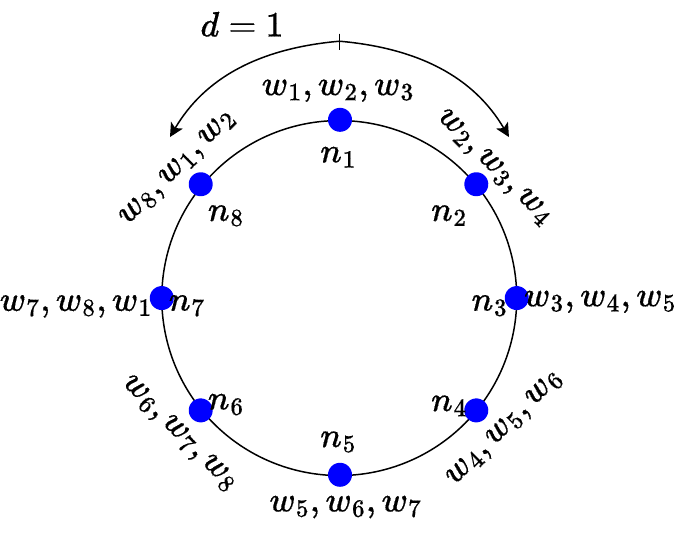}
        \caption{A file placement over a ring network with $\sN=8$ nodes, computation load $r=3$ and broadcast distance $d=1$.}
        \label{fig: MultiUnicast n=8, r=3, d=1}
\end{figure}

Among the desired IVs of node $n_i$, the packets $v_{i}^{i}, v_{i+1}^{i}, v_{i+2}^{i}$ are already known. The remaining packets are generated by different subsets of nodes, each located at varying shortest distances from node $n_i$
% \sout{The remaining packets vary in their distances from the generating nodes to node $i$} {\red WHAT DOES THIS SENTENCE MEAN??}.
We denote the set of packets cached by node $j$ and having the shortest distance of $l$ to their intended nodes as {$W_j^{(l)} = \{v_k^i: k\in\cM_j, z\in[\sN], k\in\cM_z, |i-j|=l,  |i-z| \geq l\}$}. 
% {\red I HAVE RE-DEFINED IT, CHECK IT}. 
For instance, at node $n_5$, we have $W_5^{(1)}=\{v_5^{6}, v_7^{4}\}$, $W_5^{(2)}=\{v_5^7, v_7^3\}$ and $W_5^{(3)}=\{v_5^8, v_7^2\}$.
We divide the transmission process into rounds with varying steps, where packets of different distances are sent in separate rounds. The coded packets broadcast by node $i$ in $k^{\textnormal{th}}$ step of round $j$ are denoted as $X_i^{(j,k)}$.

In the first round, each node transmits packets with distance $1$ to their destination nodes, i.e., $W_i^{(1)}$, $i\in[\sN]$. The encoding and transmission process is presented in Table~\ref{tab: MultiUnicast n=8r=3d=1round=1step=1}. For each node $i$, the coded packet $X_i^{(1,1)} = v_i^{i+1} \oplus v_{i+2}^{i-1} $ is broadcast to neighbor nodes $n_{i-1}$ and $n_{i+1}$. For example, node $n_1$ broadcasts $X_1^{(1,1)} = v_1^2 \oplus v_3^8$,   it also receives the encoded packets $X_2^{(1,1)} = v_2^3 \oplus v_4^1$ and $X_8^{(1,1)} = v_8^1 \oplus v_2^7$ which are broadcast by node $n_2$ and $n_8$, respectively. Since node $n_1$ has cached $v_2^7$ and $v_2^3$, it can decode the desired packets $v_8^1$ and $v_4^1$ by subtracting the known values from $X_2^{(1,1)} = v_2^3 \oplus v_4^1$ and $X_8^{(1,1)} = v_8^1 \oplus v_2^7$, respectively. The same decoding process applies to all other nodes.

In the second round, it will take $2$ steps of each node to transmit the packets that have a distance of $2$, i.e., $W_i^{(2)}$, $i\in[\sN]$. The encoding and transmission process is shown in Table~\ref{tab: MultiUnicast n=8r=3d=1round=2step=12}. For each node, the encoded packets of the first broadcast are generated from the packets in $W_i^{(2)}$. After receiving the coded packets from the broadcast in the first step, each node eliminates the cached content and re-encodes the remaining packets for the broadcast in the second step. For example, node $n_1$ received $X_8^{(2,1)}= v_8^2 \oplus v_2^6$ from node $n_8$ and $X_2^{(2,1)} = v_2^4 \oplus v_4^8$ from node $n_2$ in the first step. Node $n_1$ can decode $v_8^2$ and $v_4^8$ since it has cached $v_2^6$ and $v_2^4$. Then $n_1$ generates a new coded packet $X_1^{(2,2)} = v_8^2 \oplus v_4^8$ and broadcasts it in the second step.

The third round requires $3$ steps to send the packets $W_i^{(3)}$, $i\in[\sN]$, as shown in Tables~\ref{tab: MultiUnicast n=8r=3d=1round=3step=12}. The coding process follows the same principle as in previous rounds. However, we can observe that in this case, each packet in $W_i^{(3)}$ will be sent by two nodes. For example, $v_1^4$ appears in both coded packets $X_1^{(3,1)}=v_1^4\oplus v_3^6$ and $X_7^{(3,1)} = v_7^2\oplus v_1^4$, broadcast by $n_1$ and $n_7$, respectively. Therefore, the two nodes can each send half of the packets to save on transmission costs.
Finally, the NCL can be calculated as $\sT_{2} = 1+ 2 + 2.5 = 5.5$.
In this example, the decoding at each step relies solely on initially cached content. However, as we will see in the next section, from the third step onward, decoding also depends on the undesired symbols decoded in the previous two steps.
%some undesired symbols will be decoded and temporarily stored for decoding in the subsequent steps. 
% decoding depends on undesired symbols that were decoded in the previous two steps.

\begin{table}[tbp]
        \begin{minipage}[t]{0.45\linewidth}
                \centering
                \caption{Round $1$: Step $1$.}
                \label{tab: MultiUnicast n=8r=3d=1round=1step=1}
                \renewcommand\arraystretch{1.3}
                %\resizebox{\linewidth}{!}{
                \setlength{\tabcolsep}{3pt}{
                \begin{tabular}{|c|c|c|c|}
                        \hline
                        Round $1$ & \multicolumn{3}{c|}{Step $1$}  \\
                        \hline
                        Node & Broadcast & \multicolumn{2}{c|}{Receive}  \\
                        \hline
                        $n_1$ & $v_1^2 \oplus v_3^8$ & $v_8^1 \oplus v_2^7$ & $v_2^3 \oplus v_4^1$ \\
                        \hline
                        $n_2$ & $v_2^3 \oplus v_4^1$ & $v_1^2 \oplus v_3^8$ & $v_3^4 \oplus v_5^2$  \\
                        \hline
                        $n_3$ & $v_3^4 \oplus v_5^2$ & $v_2^3 \oplus v_4^1$ & $v_4^5 \oplus v_6^3$  \\
                        \hline
                        $n_4$ & $v_4^5 \oplus v_6^3$ & $v_3^4 \oplus v_5^2$ & $v_5^6 \oplus v_7^4$  \\
                        \hline
                        $n_5$ & $v_5^6 \oplus v_7^4$ & $v_4^5 \oplus v_6^3$ & $v_6^7 \oplus v_8^5$  \\
                        \hline
                        $n_6$ & $v_6^7 \oplus v_8^5$ & $v_5^6 \oplus v_7^4$ & $v_7^8 \oplus v_1^6$  \\
                        \hline
                        $n_7$ & $v_7^8 \oplus v_1^6$ & $v_6^7 \oplus v_8^5$ & $v_8^1 \oplus v_2^7$  \\
                        \hline
                        $n_8$ & $v_8^1 \oplus v_2^7$ & $v_7^8 \oplus v_1^6$ & $v_1^2 \oplus v_3^8$ \\
                        \hline
                \end{tabular}}
        \end{minipage} \hspace{9pt}
        \begin{minipage}[t]{0.45\linewidth}
                \centering
                \caption{Round $2$: Step $1$ and Step $2$.}
                \label{tab: MultiUnicast n=8r=3d=1round=2step=12}
                \renewcommand\arraystretch{1.3}
                \setlength{\tabcolsep}{3pt}{
                \begin{tabular}{|c|c|c|c|c|c|c|}
                        \hline
                        Round $2$ & \multicolumn{3}{c|}{Step $1$} & \multicolumn{3}{c|}{Step $2$}  \\
                        \hline
                        Node & Broadcast & \multicolumn{2}{c|}{Receive} & Broadcast & \multicolumn{2}{c|}{Receive}  \\
                        \hline
                        $n_1$ & $v_1^3 \oplus v_3^7$ & $v_8^2 \oplus v_2^6$ & $v_2^4 \oplus v_4^8$ & $v_8^2 \oplus v_4^8$ & $v_7^1 \oplus v_3^7$ & $v_1^3 \oplus v_5^1$  \\
                        \hline
                        $n_2$ & $v_2^4 \oplus v_4^8$ & $v_1^3 \oplus  v_3^7$ & $ v_3^5 \oplus v_5^1$ & $v_1^3 \oplus v_5^1$ & $v_8^2 \oplus v_4^8$ & $v_2^4 \oplus  v_6^2$ \\
                        \hline
                        $n_3$ & $v_3^5 \oplus v_5^1$ & $v_2^4 \oplus  v_4^8$ & $ v_4^6 \oplus v_6^2$ & $v_2^4 \oplus  v_6^2$ & $v_1^3 \oplus v_5^1$ & $v_3^5 \oplus v_7^3$ \\
                        \hline
                        $n_4$ & $v_4^6 \oplus v_6^2$ & $v_3^5 \oplus v_5^1$ & $v_5^7 \oplus v_7^3$ & $v_3^5 \oplus v_7^3$ & $v_2^4 \oplus  v_6^2$ & $v_4^6 \oplus v_8^4$ \\
                        \hline
                        $n_5$ & $v_5^7 \oplus v_7^3$ & $v_4^6 \oplus v_6^2$ & $v_6^8 \oplus v_8^4$ & $v_4^6 \oplus v_8^4$ & $v_3^5 \oplus v_7^3$ & $v_5^7 \oplus v_1^5$ \\
                        \hline
                        $n_6$ & $v_6^8 \oplus v_8^4$ & $v_5^7 \oplus v_7^3$ & $v_7^1 \oplus v_1^5$ & $v_5^7 \oplus v_1^5$ & $v_4^6 \oplus v_8^4$ & $v_6^8 \oplus v_2^6$ \\
                        \hline
                        $n_7$ & $v_7^1 \oplus v_1^5$ & $v_6^8 \oplus v_8^4$ & $v_8^2 \oplus v_2^6$ & $v_6^8 \oplus v_2^6$ & $v_5^7 \oplus v_1^5$ & $v_7^1 \oplus v_3^7$ \\
                        \hline
                        $n_8$ & $v_8^2 \oplus v_2^6$ & $v_7^1 \oplus v_1^5$ & $v_1^3 \oplus v_3^7$ & $v_7^1 \oplus v_3^7$ & $v_6^8 \oplus v_2^6$ & $v_8^2 \oplus v_4^8$ \\
                        \hline
                \end{tabular} }
        \end{minipage}
\end{table}

\begin{table}[tbp]
        \centering
        \caption{Round $3$: Step $1$, step $2$ and Step $3$.}
        \label{tab: MultiUnicast n=8r=3d=1round=3step=12}
        \renewcommand\arraystretch{1.3}
        \setlength{\tabcolsep}{3pt}{
        \begin{tabular}{|c|c|c|c|c|c|c|c|c|c|}
                \hline
                Round $3$ & \multicolumn{3}{c|}{Step $1$} & \multicolumn{3}{c|}{Step $2$} & \multicolumn{3}{c|}{Step $3$} \\
                \hline
                Node & Broadcast & \multicolumn{2}{c|}{Receive} & Broadcast & \multicolumn{2}{c|}{Receive} & Broadcast & \multicolumn{2}{c|}{Receive} \\
                \hline
                $n_1$ & $v_1^4 \oplus v_3^6$ & $v_8^3 \oplus v_2^5$ & $v_2^5 \oplus v_4^7$ & $v_8^3 \oplus v_4^7$ & $v_7^2 \oplus v_3^6$ & $v_1^4 \oplus v_5^8$ & $v_7^2 \oplus v_5^8$ & $v_6^1 \oplus v_4^7$ & $v_8^3 \oplus v_6^1$  \\
                \hline
                $n_2$ & $v_2^5 \oplus v_4^7$ & $v_1^4 \oplus v_3^6$ & $v_3^6 \oplus v_5^8$ & $v_1^4 \oplus v_5^8$ & $v_8^3 \oplus v_4^7$ & $v_2^5 \oplus  v_6^1$  & $v_8^3 \oplus v_6^1$ & $v_7^2 \oplus v_5^8$ & $v_1^4 \oplus v_7^2$ \\
                \hline
                $n_3$ & $v_3^6 \oplus v_5^8$ & $v_2^5 \oplus v_4^7$ & $v_4^7 \oplus v_6^1$ & $v_2^5 \oplus v_6^1$ & $v_1^4 \oplus v_5^8$ & $v_3^6 \oplus v_7^2$ & $v_1^4 \oplus v_7^2$ & $v_8^3 \oplus v_6^1$ & $v_2^5 \oplus v_8^3$ \\
                \hline
                $n_4$ & $v_4^7 \oplus v_6^1$ & $v_3^6 \oplus v_5^8$ & $v_5^8 \oplus v_7^2$ & $v_3^6 \oplus v_7^2$ & $v_2^5 \oplus  v_6^1$ & $v_4^7 \oplus v_8^3$ & $v_2^5 \oplus v_8^3$ & $v_1^4 \oplus v_7^2$ & $v_3^6 \oplus v_1^4$ \\
                \hline
                $n_5$ & $v_5^8 \oplus v_7^2$ & $v_4^7 \oplus v_6^1$ & $v_6^1 \oplus v_8^3$ & $v_4^7 \oplus v_8^3$ & $v_3^6 \oplus v_7^2$ & $v_5^8 \oplus v_1^4$ & $v_3^6 \oplus v_1^4$ & $v_2^5 \oplus v_8^3$ & $v_4^7 \oplus v_2^5$ \\
                \hline
                $n_6$ & $v_6^1 \oplus v_8^3$ & $v_5^8 \oplus v_7^2$ & $v_7^2 \oplus v_1^4$ & $v_5^8 \oplus v_1^4$ & $v_4^7 \oplus v_8^3$ & $v_6^1 \oplus v_2^5$ & $v_4^7 \oplus v_2^5$ & $v_3^6 \oplus v_1^4$ & $v_5^8 \oplus v_3^6$ \\
                \hline
                $n_7$ & $v_7^2 \oplus v_1^4$ & $v_6^1 \oplus v_8^3$ & $v_8^3 \oplus v_2^5$ & $v_6^1 \oplus v_2^5$ & $v_5^8 \oplus v_1^4$ & $v_7^2 \oplus v_3^6$ & $v_5^8 \oplus v_3^6$ & $v_4^7 \oplus v_2^5$ & $v_6^1 \oplus v_4^7$ \\
                \hline
                $n_8$ & $v_8^3 \oplus v_2^5$ & $v_7^2 \oplus v_1^4$ & $v_1^4 \oplus v_3^6$ & $v_7^2 \oplus v_3^6$ & $v_6^1 \oplus v_2^5$ & $v_8^3 \oplus v_4^7$ & $v_6^1 \oplus v_4^7$ & $v_5^8 \oplus v_3^6$ & $v_7^2 \oplus v_5^8$ \\
                \hline
        \end{tabular}}
\end{table}

\subsection{General Scheme}
\label{sec: all-to-alls_general}
% In this section, we introduce our general scheme. 
Initially, the input files are cached by the nodes based on the cyclic placement along the ring network. Specifically, node $n_i$ caches input files $w_j$ if $i \leq j \leq i+r-1$, i.e.,
\begin{align}
        \cM_i = \{w_i, \ldots, w_{i+r-1}\}.\nonumber
\end{align}
After the computation of local mapping, node $n_i$ knows the IVs computed from the files in $\cM_i$, i.e., $\{v_n^j: w_n\in\cM_i, j\in[\sN]\}$.
We first discuss the cases with the computation load $r\in\{2,\ldots, \sN-1\}$ and broadcast distance $d=1$.% where the transmission strategy is presented in Algorithm \ref{alg:MUoverRing}.
Subsequently, we describe modifications to the scheme for cases with $d \geq 2$ or $r=1$. 

\subsubsection{Proposed scheme for $d=1$} 
For $r\in\{2,\ldots, \sN-1\}$ and $d=1$, the IVs are carefully delivered based on their distance to the intended node, recalling that 
% {\red PREVIOUSLY, YOU USE $W_j^{(l)}$.... }
\begin{align*}
        % W_i^{(l)} \!=\! \{v_k^j: k\in\!\cM_i,  z\in[\sN],l=|j\!-\!i|\leq |j\!-\!z|, k\in\!\cM_z \}.
        {W_j^{(l)} \!\!=\!\! \{v_k^i\!:\! k\!\in\!\cM_j, z\in[\sN], k\!\in\!\cM_z, |i-j|=l,  |i-z| \geq l\}.}
\end{align*}
As described in Section~\ref{sec: MUexa}, the transmitted IVs of node $i$ can be categorized into different sets $W_i^{(1)}, \ldots, W_i^{(\lceil\frac{\sN-r}{2}\rceil)}$.
%  according to their distance to the intended node, 
{Since the packets follow the cyclic placement, the desired packets are located at most $\left\lceil \frac{\sN - r}{2} \right\rceil$ hops away from node $n_i$.} 

The transmissions are performed in rounds, with packets at different distances being sent in separate rounds.

At the $k^{th}$ step of round $j$ { where $k\in [j]$ and $j\in [\lceil\frac{\sN-r}{2}\rceil]$}, node $n_i$ generates the coded packets 
\begin{align}
        X_i^{(j,k)} = v_{i-(k-1)}^{i-(k-1)+j} \oplus v_{i+(r-1)+(k-1)}^{i+(k-1)-j}, \label{eq:MU-Xi}
\end{align}
and broadcasts it to nodes $n_{i-1}$ and $n_{i+1}$. Meanwhile, $n_i$ receives the coded packets {$X_{i-1}^{(j,k)} = v_{i-k}^{i-k+j} \oplus v_{i+r+k-3}^{i+k-2-j}$ and $X_{i+1}^{(j,k)} = v_{i-k+2}^{i-k+j+2} \oplus v_{i+r+k-1}^{i+k-j}$ from nodes $n_{i-1}$ and $n_{i+1}$, respectively. Collectively, these are denoted by 
\begin{align*}
    Y_i^{(j,k)} = (X_{i-1}^{(j,k)}, X_{i+1}^{(j,k)}).
\end{align*}}
For $k\in\{1,2\}$, $v_{i+r+k-3}^{i+k-2-j}$ and $v_{i-k+2}^{i-k+j+2}$ are already cached by node $n_i$. Thus, node $n_i$ can easily decode $v_{i-k}^{i-k+j}$ and $v_{i+r+k-1}^{i+k-j}$. These decoded packets $v_{i+r+k-1}^{i+k-j}$ and $v_{i-k}^{i-k+j}$ are added to the cache of node $n_i$ for future use in generating transmitted packets in the next step and for decoding in the $(k+2)^{\text{th}}$ step. 
During the repeated and iterative processing, for $k>2$, $v_{i+r+k-3}^{i+k-2-j}$ and $v_{i-k+2}^{i-k+j+2}$, can be found in $n_i$'s cache and used for decoding. To save storage space, $n_i$ can remove $v_{i+r+k-3}^{i+k-2-j}$ and $v_{i-k+2}^{i-k+j+2}$ from its cache, after the $k^{\textnormal{th}}$ step. In other words, it requires temporarily storing the decoded packets in the $k^{\text{th}}$ step to facilitate the decoding in the $(k+2)^{\text{th}}$.
% The scheme requires temporarily storing the results obtained from decoding in the previous two steps to facilitate further decoding.  
Repeating the above process until the $j^{\textnormal{th}}$ step of round $j$, at which point node $n_i$ receives 
\begin{align*}
    Y_i^{(j, j)}  = (X_{i-1}^{(j, j)}, X_{i+1}^{(j, j)}) 
\end{align*}
where $X_{i-1}^{(j, j)} = v_{i-j}^{i} \oplus v_{i+r+j-3}^{i-2}$ and $X_{i+1}^{(j, j)} = v_{i-j+2}^{i+2} \oplus v_{i+r+j-1}^{i}$.
The packets $v_{i-j}^{i}$ and $v_{i+r+j-1}^{i}$ finally reach their destined node $n_i$ and are successfully decoded. In this scheme, $\sTe = \lceil\frac{\sN-r}{2}\rceil$ round of transmissions are required, with $j$ steps in $j^{\textnormal{th}}$ round. Therefore, we can calculate the NCL as 
\begin{align}
        \sT_{\tn{2-cyc}}(r,1) &= \left(1 + 2 + \ldots + \Big\lceil\frac{\sN-r}{2}\Big\rceil \right)\cdot \frac{\sN\sB_{2}}{\sN\sB_{2}} \nonumber\\
        % &= O\left(\frac{(\sN-r)^2}{8}+\frac{\sN-r}{4}\right). \label{eq: T(r,1)}
        % &= O\left(\frac{\sN}{4}(\frac{\sN}{2}\!\!-\!\!r) \!+\! \frac{(\sN-r)}{4}\!+\! \frac{r^2}{8}\right) \nonumber\\ 
        &{= \frac{1}{2}\left(\left\lceil\frac{\sN-r}{2}\right\rceil\right)^2 + \frac{1}{2}\left\lceil\frac{\sN-r}{2}\right\rceil}. \label{eq: T(r,1)}
        % \frac{\sN}{4}(\frac{\sN}{2}\!\!-\!\!r) \!+\! \frac{(\sN-r)}{4}\!+\! \frac{r^2}{8}
\end{align}
% {where $m=\lceil\frac{\sN-r}{2}\rceil$.}
{The pseudo code of the proposed scheme can be found in Algorithm~\ref{alg:MUoverRing}.} 
% {\red I PUT IT HERE, RIGHT???}

Note that for the packets in $W_i^{\lceil\frac{\sN-r}{2}\rceil}=\{v_i^{i+\lceil\frac{\sN-r}{2}\rceil}, v_{i+r-1}^{i-\lceil\frac{\sN-r}{2}\rceil}\}$, the packet $v_i^{i+\lceil\frac{\sN-r}{2}\rceil}$ is transmitted by two nodes $\{n_i, n_{i-r+1}\}$ and $v_{i+r-1}^{i-\lceil\frac{\sN-r}{2}\rceil}$ is transmitted by two nodes $\{n_i, n_{i+r-1}\}$ during the first step of round $\lceil\frac{\sN-r}{2}\rceil$ when $\frac{\sN-r}{2} \notin \mathbb{N}^{+}$. To reduce the communication costs, each pair of nodes can send complementary halves of the same packet. 
% {\red \{THIS COINCIDES WITH WHICH PART IN THE ALGORITHM??? 
{As its impact on the results becomes negligible for large $\sN$, it is omitted from the pseudocode for brevity.}

\begin{algorithm}
        \caption{Transmission for all-to-all over Ring Networks ($d=1, r>1$)}\label{alg:MUoverRing}
        \begin{algorithmic}[1]
        %\STATE 
        \STATE $n_i$ caches $\{{v}_i^j, \ldots, {v}_{i+r-1}^j: j\in[\sN\}$;

        \FOR{Round $j=1,\ldots, \lceil\frac{\sN-r}{2}\rceil$}
                %\STATE Round $j$:
                \FOR{Step $k=1,2,\ldots, j$}
                %\STATE Step $k$:
                \STATE Node $n_i$ transmits $X_i^{(j,k)}$ to nodes $\{n_{i-1}, n_{i+1}\}$.
                                %$$ X_i^{(j)}(k) = v_{i-(k-1)}^{(i-(k-1)+j)} \oplus v_{i+(r-1)+(k-1)}^{(i+(r-1)+(k-1)-j)} $$
                                $$ X_i^{(j, k)} = v_{i-(k-1)}^{i-(k-1)+j} \oplus v_{i+(r-1)+(k-1)}^{i+(k-1)-j} .$$
                        %to nodes $\{n_{i-1}, n_{i+1}\}$.
                % \STATE Node $n_i$ receives $ Y_i^{(j,k)}=$$( v_{i-k}^{i-k+j} \oplus v_{i+r+k-3}^{i+k-2-j} ,  v_{i-k+2}^{i-k+j+2} \oplus v_{i+r+k-1}^{i+k-j} )$,
                \STATE Node $n_i$ receives $ Y_i^{(j,k)}=$$( X_{i-1}^{(j,k)} ,  X_{i+1}^{(j,k)} )$,
                      %  \begin{align}
                       %         Y_i^{(j,k)}  \left\{ \begin{array}{c}\vspace{5pt} 
                        %                v_{i-k}^{i-k+j} \oplus v_{i+r+k-3}^{i+k-2-j} \\
                          %              v_{i-k+2}^{i-k+j+2} \oplus v_{i+r+k-1}^{i+k-j} 
                          %      \end{array} \right.
                        %\end{align}
                % \STATE Decode $v_{i-k}^{(i-k+j)}$ and $v_{i+r+k-1}^{(i+k-j)}$.
                    and decodes $v_{i-k}^{(i-k+j)}$ and $v_{i+r+k-1}^{(i+k-j)}$.
                \vspace{5pt} 
                \STATE Cache update: Delete $v_{i-k+2}^{(i-k+j+2)}$ and $v_{i+r+k-3}^{(i+k-2-j)}$.
                \vspace{5pt}
                \STATE Cache update: Add $v_{i-k}^{(i-k+j)}$ and $v_{i+r+k-1}^{(i+k-j)}$. \label{alg-BuffAdd}
                \ENDFOR
        \ENDFOR
        \end{algorithmic}
\end{algorithm}

% {\red KAI STOPPED HEREHEREHEREHEREHERE}
\subsubsection{Modification for $1 < d\leq 2(r-1)$} 
The number of broadcasts can be further reduced when the nodes have a larger broadcast distance. Based on the transmission strategy for $d=1$, here is the modified version for $d\leq 2(r-1)$. Unlike the case of $d=1$ where the $3$ adjacent nodes perform reverse carpooling, here we allow nodes that are distance $d$ or $r-1$ apart to perform reverse carpooling. The paths of information flow are demonstrated in Fig.~\ref{fig: MUGen(line)}. 

Define $d_1 = \min\{d,r-1\}$. For $i\in[\sN]$, symbols in $\{W_i^{(1)}, \ldots, W_i^{(d_1)}\}$ can reach the intended nodes in a single broadcast. 
In the first step of round $j>d_1$, the nodes $(n_{i-d_1}, n_i, n_{i+d_1})$ form a reverse carpooling topology. Node $n_i$ broadcasts 
\begin{align}
        X_i^{(j,1)} = v_{i}^{i+j} \oplus v_{i+r-1}^{i-j} \label{eq: X_i(1)-d>1}
\end{align}
to nodes $n_{i-d_1}$ and $n_{i+d_1}$, and receives 
\begin{align}
        Y_i^{(j,1)} = \left( 
                X_{i-d_1}^{(j,1)} = v_{i-d_1}^{i-d_1+j} \oplus v_{i-d_1+r-1}^{i-d_1-j},
                X_{i+d_1}^{(j,1)} = v_{i+d_1}^{i+d_1+j} \oplus v_{i+d_1+r-1}^{i+d_1-j}  \right)
\end{align}
Node $n_i$ can successfully decode $v_{i-d_1}^{i-d_1+j}$ and $v_{i+d_1+r-1}^{i+d_1-j}$ since it already caches $v_{i-d_1+r-1}^{i-d_1-j}$ and $v_{i+d_1}^{i+d_1+j}$. It then broadcasts the mixture
\begin{align*}
        X_i^{(j,2)} = v_{i-d_1}^{i-d_1+j} \oplus v_{i+d_1+r-1}^{i+d_1-j}
\end{align*}
in the second step. 
Starting from the second step, nodes $(n_{i-d}, n_i, n_{i+d})$ form a reverse carpooling topology. Node $n_i$ receives 
\begin{align}
        Y_i^{(j,2)}  \left(  X_{i-d}^{(j,2)} = v_{i-d-d_1}^{i-d-d_1+j} \oplus v_{i - d + d_1+r-1}^{i-d+d_1-j}, X_{i+d}^{(j,2)} = v_{i+d-d_1}^{i+d-d_1+j} \oplus v_{i+d+d_1+r-1}^{i+d+d_1-j} \label{eq: Y_i(2)-d>1} \right) .
\end{align}
Observe that $v_{i - d + d_1+r-1}^{i-d+d_1-j}$ and $v_{i+d-d_1}^{i+d-d_1+j}$ are cached by node $n_i$. The transmission can proceed repeatedly and iteratively in a manner analogous to the case where $d = 1$.
Similarly, it requires temporarily storing the decoded packets in $k^{\text{th}}$ step to facilitate decoding in $(k+2)^{\text{th}}$.
% obtained from decoding in the previous two steps at each step to facilitate subsequent decoding. 

At the end of the $\left(\lceil\frac{j-d_1}{d}\rceil\right)^{\textnormal{th}}$ step, the packets will reach the nodes near their destination. In the final step, the nodes generate the coded packet based on the packets coming from opposite directions and perform regular broadcasts to their intended nodes. 
To enable successful decoding in the final step, when nodes $(n_{i-d}, n_i, n_{i+d})$ (or $(n_{i-d_1}, n_i, n_{i+d_1})$) execute reverse carpooling in earlier steps, the nodes between $n_{i-d}$ and $n_i$, as well as the nodes between $n_i$ and $n_{i+d}$, need to decode the new packets lately flowing through them, though they are not required to forward these packets. This ensures that the intended node in the final step can decode the desired packets by eliminating the undesired ones.
The process is demonstrated in Fig.~\ref{fig: MUGen(line)-finalstep}.
% {\color{blue}new fig}
In the cases $1<r$ and $1<d \leq 2(r-1)$, we can calculate the NCL as 
\begin{align}
        \lefteqn{\sT_{\tn{2-cyc}}(r,d)} \nonumber\\
        &= \underbrace{1 +\ldots+ 1}_{d_1} + \underbrace{2+ \ldots + 2}_{d}+ \ldots  +\! \underbrace{\left\lceil\frac{\left\lceil\frac{\sN-r}{2}\right\rceil\!-\!d_1}{d}\right\rceil \!+\! 1 \!+\! \ldots \!+\! \left\lceil\frac{\left\lceil\frac{\sN-r}{2}\right\rceil\!-\!d_1}{d}\right\rceil \!+\! 1}_{\left(\lceil\frac{\sN-r}{2}\rceil-d_1\right)_{\bmod d}} \nonumber\\
        &{= \frac{d}{2}\! \flxd{\ceNr - d_1}^2 \!+\! \frac{3d}{2}\!\flxd{\ceNr \!-\! d_1}\! \!+\! d_1 \!+\! \left(\flxd{\ceNr \!-\! d_1} \!+\! 2\right)\left(\ceNr \!-\! d\right)_d}\nonumber \\
        &{=\left\{\begin{array}{cc}\vspace{5pt} 
                \frac{d}{2}\! \flxd{\ceNr - r + 1}^2 \!+\! \frac{3d}{2}\flxd{\ceNr - r + 1} \!+\! r \!-\! 1 \!+\! \left(\flxd{\ceNr - r + 1} \!+\! 2\right)\left(\ceNr \!-\! d\right)_d &  \substack{d_1 = r-1,\\ 1<r,}\\
                \frac{d}{2}\!\flxd{\ceNr}^2 \!+\! \frac{d}{2}\! \flxd{\ceNr} \!+\! \left(\flxd{\ceNr} + 1\right) \left(\ceNr \!-\! d\right)_d&  \substack{d_1 = d,\\ 1<r,}
        \end{array}\right. } \label{eq: T(1<r,1<d<2r-1)}
\end{align}
{where $(x)_d$ denotes $x \bmod d$ for $x \in \mathbb{Z}$.}

\begin{figure}[t]
        \centering
        \includegraphics[width=0.55\linewidth]{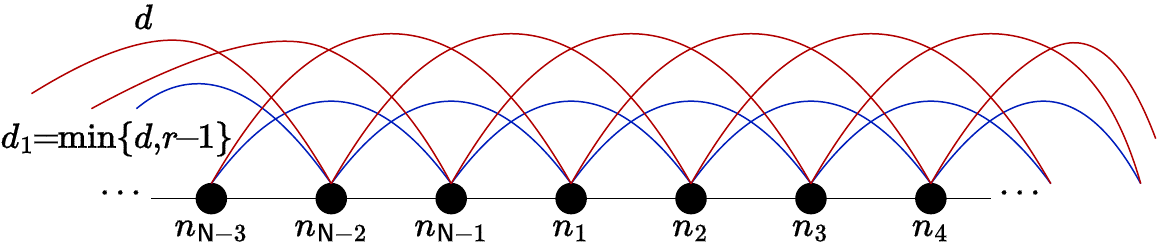}
        \caption{From round $d_1+1$ to the last round of transmission when $1<d\leq 2(r-1)$: any three nodes connected by two adjacent blue solid lines form a \emph{reverse carpooling} topology, where the first step of transmission takes place within this topology; any three nodes connected by two adjacent red solid lines also form a \emph{reverse carpooling} topology, where the transmissions from the second step to the second-to-last step take place within this topology.}
        \label{fig: MUGen(line)}
\end{figure}
\begin{figure}[t]
        \centering
        \subfloat[{\footnotesize Routing to near nodes of destination}]{
                \includegraphics[width=0.5\linewidth]{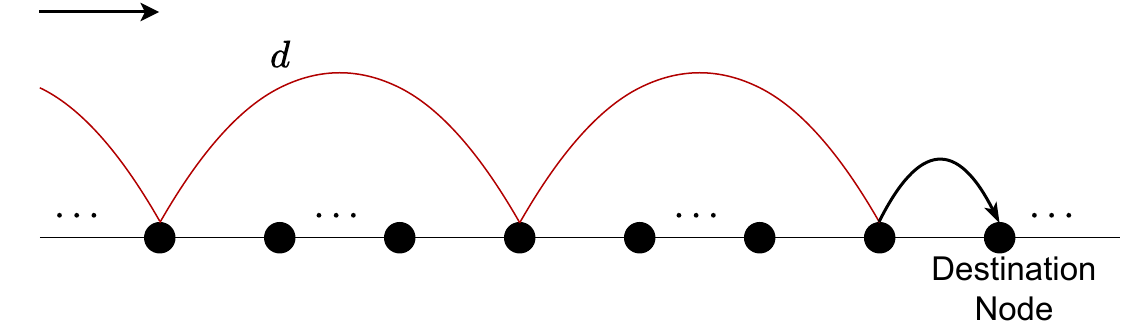}
        }
        \\
        \subfloat[{\footnotesize Broadcast in last step}]{
        \includegraphics[width=0.5\linewidth]{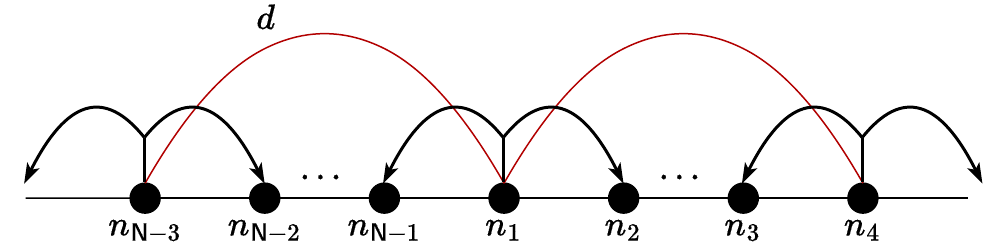}
        }
        \caption{(a) At the end of the second-to-last step in any transmission round, the packets reach the nodes near their destination nodes. (b) At the last step of transmission in any round, node $n_i$ performs a normal multicast to one of pairs of two nodes $\left(n_{i-1}, n_{i+1}\right), \ldots, \left(n_{i-\lceil\frac{d-1}{2}\rceil},n_{i+\lceil\frac{d-1}{2}\rceil}\right)$.}
        \label{fig: MUGen(line)-finalstep}
\end{figure}

\subsubsection{Modification for $2r-1\leq d$} 
We may want to generate coded packets as in \eqref{eq: X_i(1)-d>1} for the first step of transmission when $r=1$ (i.e., let $d_1=0$). However, we can observe that reverse carpooling of packets at the first step is infeasible here,
as there are no redundant contents between any pair of nodes that could be exploited for decoding. 
For $r>1$, the same problem occurs when decoding the packets in \eqref{eq: Y_i(2)-d>1}. Specifically, $v_{i+d-d_1}^{i+d-d_1+j}$ and $v_{i - d + 2d_1}^{i-d+d_1-j}$ are not cached by node $n_i$. A solution is to add an additional transmission step and modify the transmission content in the first step as follows. Let node $n_i$ transmit
\begin{align*}
        X_i^{(j, 0)} = v_{i+r-1}^{i-j}, \text{ and } X_i^{(j,1)} = v_{i}^{i+j}
\end{align*}
to nodes $n_{i-d}$ and $n_{i+d}$, respectively. For $i\in[\sN], j\in[\lceil\frac{\sN-r}{2}\rceil]$, $X_i^{(j, 0)}$ 
% are additional transmission before the first step. The remaining steps of the process can proceed as described previously. 
{constitute the additional step before the first regular transmission. These packets provide the necessary side information for decoding in the first step, enabling the pipeline to proceed according to the same decoding progression as described previously.}
Therefore, in this case, we can calculate the NCL as 
\begin{align}
        \sT_{\tn{2-cyc}}(r,d) 
        &= \underbrace{1 +\ldots+ 1}_{d} + \underbrace{2+ \ldots + 2}_{d}+ \ldots  +\! \underbrace{ \left\lceil\frac{\left\lceil\frac{\sN-r}{2}\right\rceil}{d}\right\rceil\!+\! \ldots \!+\!  \left\lceil\frac{\left\lceil\frac{\sN-r}{2}\right\rceil}{d}\right\rceil }_{ \lceil\frac{\sN-r}{2}\rceil_{\bmod d} } + \left\lceil\frac{\sN-r}{2}\right\rceil \nonumber\\
        & = {\frac{d}{2}\flxd{\ceNr}^2 + \frac{d}{2}\flxd{\ceNr} + \flxd{\ceNr}\left(\ceNr\right)_d + \ceNr}. \label{eq: T(2r-1<d)}
        % &{=\frac{d}{2}\left(\left\lceil\frac{\sN\!-\!r}{2d}\right\rceil\right)^2 \!-\! \frac{d}{2}\left\lceil\frac{\sN\!-\!r}{2d}\right\rceil \!+\! \left\lceil\frac{\sN\!-\!r}{2d}\right\rceil \left(\left\lceil\frac{\sN\!-\!r}{2}\right\rceil\right)_d \!+\! \left\lceil\frac{\sN\!-\!r}{2}\right\rceil},
\end{align}
{where $(x)_d$ denotes $x \bmod d$ for $x \in \mathbb{Z}$.}
From \eqref{eq: T(r,1)}, \eqref{eq: T(1<r,1<d<2r-1)} and \eqref{eq: T(2r-1<d)}, we can conclude that the achievable NCL of the proposed scheme is as stated in Theorem \ref{Theo: allunicast}.

\section{Concluding Remark and Future Directions}\label{sec: conclusion}
This paper investigates the tradeoff between NCL, computation load, and broadcast distance of the coded distributed systems over the ring network. For all-gather computing, we exactly characterized an information-theoretic tradeoff between NCL, computation load, and broadcast distance. For all-to-all computing, we proposed a coded transmission scheme that utilizes a simple network coding method for efficient broadcasting. We further prove that the proposed scheme is asymptotically optimal under cyclic data placement through information-theoretic converse arguments. We found that the coded transmission gain depends on the connectivity among nodes instead of redundant computation load in the network topology of the ring.  Notably, the optimal tradeoff between NCL, computation load, and broadcast distance under arbitrary file placement remains an open problem for future work. Besides, extensions to other network topologies as follows, such as 2-D torus networks, satellite networks, are also worth investigating.
{
\subsection{Extension to $2$-d Torus Networks}
        The ring network studied in this work is a one-dimensional $k$-nearest-neighbor circulant graph. The proposed framework can be extended to other structured topologies, such as two-dimensional torus, provided that the broadcast capability is defined in a compatible manner. In particular, consider a 2-D torus network where each node performs distance-limited broadcast along one dimension at a time, with the broadcast distance defined analogously to the ring model. In this case, a natural extension is a dimension-wise nested scheme.
        % Consider a 2-D torus, if each node is allowed to perform distance-limited broadcast along one dimension at a time, with broadcast distance defined analogously to the ring model, then a natural extension is a dimension-wise nested scheme. 
        For the all-gather task, one may first apply the proposed ring-based coded transmission independently along one dimension (e.g., rows) to aggregate intermediate values within each row, and then apply the same scheme along the other dimension (e.g., columns) to complete global aggregation. For the all-to-all task, a similar nesting applies, where intermediate values are first delivered to their target rows and then delivered to their final destinations along the second dimension using the proposed schedules within each sub-ring. 
        Under this nested construction, reverse-carpooling opportunities arise independently along each dimension, preserving the key insights established for ring networks. If the broadcast model differs substantially from that assumed above, new coding and scheduling designs may be required, which is a worthy direction for future work.
}
{
\subsection{Relevance to Satellite Orbital-Plane Networks.}
        The ring abstraction studied in this work is also relevant to satellite constellations, where satellites within a single orbital plane form a ring via inter-satellite links (ISLs) \cite{radhakrishnan2016survey, ekici2001distributed}. In this setting, $\sN$ corresponds to the number of satellites per plane, $d$ represents the hop budget of ISL communication within a scheduling window, and $r$ reflects redundancy enabled by replicated sensing or preprocessing. The all-gather and all-to-all tasks naturally arise from plane-wide aggregation (e.g., ephemeris dissemination or onboard model aggregation) and directional exchanges (e.g., task offloading or cooperative calibration), respectively. The reduction in normalized communication load implies fewer ISL transmissions, which is particularly relevant under half-duplex and slotted operation. The investigations under more complex broadcast or dynamic connectivity models are interesting future directions.
}

\begin{appendices}

    \section{A Lower Bound on NCL for all-gather Computing}
    % \subsection{Proof of the lower bound in Theorem \ref{Theo: aa_uplow}}
    \label{sec: sub-prof_sintasks_lb}
    The main idea follows Lemma 1 in \cite{fragouli2008efficient}. For each IV, there are $\sN-r$ receivers to cover, and it can reach $2d$ receivers through one broadcast transmission. Formally, we have the information-theoretic representation as follows.
        From \eqref{eq: Dec-alltoall} we have 
        \begin{align}
                H\left(\left( V_1,\ldots,V_{\sN}\right) \Big| \cM_i, \left(Y_{i}^{(1)},\ldots,Y_{i}^{(\sT^{e})}\right)\right) = 0.
        \end{align}
        From the definition of mutual information, we have
        \begin{align}
                H\left(V_1,\ldots,V_{\sN}\right)  &= I\left( \left(V_1,\ldots,V_{\sN}\right) ; \cM_i, \left( Y_{i}^{(1)},\ldots,Y_{i}^{(\sT^{e})}\right)\right) \nonumber\\
                &= H\left( \cM_i, \left(Y_{i}^{(1)},\ldots,Y_{i}^{(\sT^{e})}\right) \right) - H\left( \cM_i, \left(Y_{i}^{(1)},\ldots,Y_{i}^{(\sT^{e})}\right) \Big| \left(V_1,\ldots,V_{\sN}\right)\right) \nonumber\\
                &\leq H\left(\cM_i\right) + H\left( Y_{i}^{(1)},\ldots,Y_{i}^{(\sT^{e})} \right)
        \end{align}
        Then, we have
        \begin{align}\label{eq: entropyY}
               %&\lefteqn{H\left( \{Y_{i}^{(1)},\ldots,Y_{i}^{(\sT^{e})}\} \right) } \nonumber\\
                H\left( Y_{i}^{(1)},\ldots,Y_{i}^{(\sT^{e})} \right) &\geq H\left( V_1,\ldots,V_{\sN} \right) - H\left(\cM_i\right) \nonumber \\
                & = (\sN-|\cM_i|)\sB_{1}
        \end{align}
        Besides, 
        \begin{align}\label{eq: bitsY}
                H\left( Y_{i}^{(1)},\ldots,Y_{i}^{(\sT^{e})} \right) \leq \sum_{t=1}^{\sT^{e}}\sum_{k=i-d}^{i+d}l_{k}^{(t)}
        \end{align}
        Combining \eqref{eq: entropyY} and \eqref{eq: bitsY}, we have
        \begin{align}
                \sum_{t=1}^{\sT^{e}}\sum_{k=i-d}^{i+d}l_{k}^{(t)} \geq (\sN-|\cM_i|)\sB_{1}
        \end{align}
        Summing up the above inequality over $i\in \{1,\ldots,\sN\}$, we obtain
        \begin{align}
                \sum_{i=1}^{\sN}\sum_{t=1}^{\sT^{e}}\sum_{k=i-d}^{i+d}l_{k}^{(t)} \geq (\sN^2-\sum_{i=1}^{\sN}|\cM_i|)\sB_{1},
        \end{align} 
        i.e., 
        \begin{align}
                2d\sum_{t=1}^{\sT^{e}}\sum_{k=1}^{\sN}l_{k}^{(t)} \geq \sN(\sN-r)\sB_{1}.
        \end{align}
        Then it is readily shown that 
        \begin{align}
                \frac{\sum_{t=1}^{\sT^{e}}\sum_{k=1}^{\sN}l_{k}^{(t)}}{\sN\sB_{1}} \geq \frac{\sN-r}{2d}.
        \end{align}
        This completes the proof.

    \section{A Lower Bound on NCL for all-to-all Computing $(r\leq \lceil\frac{\sN}{2}\rceil-1)$}
    % \subsection{Proof of the Theorem \ref{Theo: au_con-cyc}}
    \label{sec: all-to-alls_lb_placement}
    % Here we present the proof of the Theorem \ref{Theo: au_con-cyc}. 
    Consider a cut as Fig.~\ref{fig: ConvProvNew} shown; {we divide the nodes into two sets, $\mathcal{A}_j$ and the complement $\bar{\mathcal{A}}_j$, where the nodes $\{n_{j-d},\ldots,n_{j-1}\}$ and $\{n_{j+s},\ldots,n_{j+s+d-1}\}$ belong to ${\bar{\cal A}_j}$ have direct connection with ${\cal A}_j$. 
    Define the set of files cached by nodes in ${\cal A}_j$ as 
    \begin{align}
        \cM_{{\cal A}_j} \triangleq \underset{{k\in{\cal A}_j}}{\cup} {\cal M}_k \subseteq \{w_1, \ldots, w_{\sN}\},
    \end{align}
    and its complement as 
    \begin{align}
        \bar{\cM}_{{\cal A}_j} \triangleq \{w_1, \ldots, w_{\sN}\}\backslash \cM_{{\cal A}_j}
    \end{align}
    which has the size of $|\bar{\cM}_{{\cal A}_j}| = m_j$.
    Additionally, we denote the set of IVs that are required by nodes in ${\cal A}_j$ but cannot be computed locally inside ${\cal A}_j$ as
    \begin{align}
        V({\cal A}_j) = \{v_i^k: w_i \in \bar{\cM}_{{\cal A}_j}, k\in{\cal A}_j\},
    \end{align} 
    and denote $\{X_{j-d}^{\sTe}, \ldots, X_{j-1}^{\sTe}\}$ and $\{X_{j+s}^{\sTe}, \ldots, X_{j+s+d-1}^{\sTe}\}$ as $X_{j-}^{\sTe}$ and $X_{j+}^{\sTe}$, respectively.
        \begin{figure}[tbp]
            \centering
            \includegraphics[width=0.25\linewidth]{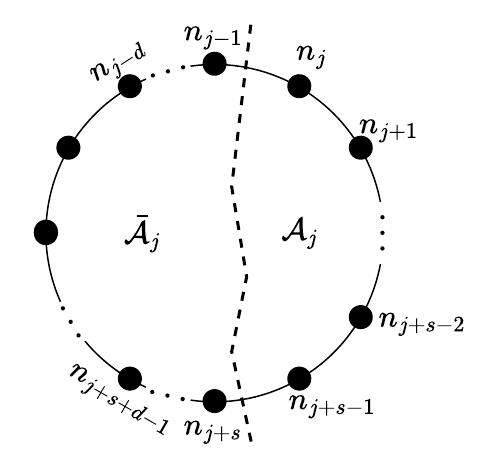}
            \caption{A cut of the nodes. The nodes are divided into two complementary subsets, $\cal A_j$ and $\bar{\cal A}_j$, separated by a dotted line.}
            \label{fig: ConvProvNew}
        \end{figure}
     We have the following relations
     \begin{align}
        H(V({\cal A}_j)) &= H(V({\cal A}_j)| \cM_{{\cal A}_j}) \label{eq: ConP-inde2}\\
        &= I(X_{j-}^{\sTe}, X_{j+}^{\sTe}; V({\cal A}_j)| \cM_{{\cal A}_j}) + H( V({\cal A}_j)| X_{j-}^{\sTe}, X_{j+}^{\sTe}, \cM_{{\cal A}_j} ) \label{eq: ConP-DeMI2}\\
        &= I(X_{j-}^{\sTe}, X_{j+}^{\sTe}; V({\cal A}_j)| \cM_{{\cal A}_j}) \label{eq: ConP-Fano2}\\
        &= H(X_{j-}^{\sTe}, X_{j+}^{\sTe}| \cM_{{\cal A}_j}) - H(X_{j-}^{\sTe}, X_{j+}^{\sTe}|V({\cal A}_j), \cM_{{\cal A}_j}) \nonumber \\
        &\leq H(X_{j-}^{\sTe}, X_{j+}^{\sTe} ) \label{eq: ConP-lgeqa2}\\
        &= \sum_{i=1}^{d}\sum_{t=1}^{\sTe}(l_{j-i}^{(t)} + l_{j+s+i-1}^{(t)}), \label{eq: ConP-lgeqa3}
     \end{align}
     where \eqref{eq: ConP-inde2} holds because of the independence between packets, \eqref{eq: ConP-DeMI2} follows the definition of the mutual information, \eqref{eq: ConP-Fano2} comes from that $H( V({\cal A}_j)| Y_{{\cal A}_j}^{\sTe}, \cM_{{\cal A}_j} ) = 0$ holds because of the decoding constraint and $Y_{{\cal A}_j}^{\sTe}$ are functions of $(X_{j-}^{\sTe}, X_{j+}^{\sTe}, \cM_{{\cal A}_j})$, and \eqref{eq: ConP-lgeqa2} holds because conditioning reduces entropy, and \eqref{eq: ConP-lgeqa3} follows the definition of $X_{j-}^{\sTe}$ and $X_{j+}^{\sTe}$. Besides, 
     \begin{align}
        H(V({\cal A}_j)) = |\bar{\cM}_{{\cal A}_j}| \cdot |{\cal A}_j| \cdot \sB_{2} = m_j \cdot s \cdot \sB_{2},
     \end{align}
     which implies that
     \begin{align}
        \sum_{i=1}^{d}\sum_{t=1}^{\sTe}(l_{j-i}^{(t)} + l_{j+s+i-1}^{(t)}) \geq m_j \cdot s \cdot \sB_{2} \label{eq: ConP-lgeqacut}.
     \end{align}}
     Taking the partitions $\{{\cal A}_j, \bar{\cal A}_j\}$ as the original version, we can generate the new cut by moving some nodes. We move the node $n_{j}$ from ${\cal A}_j$ to ${\bar{\cal A}_j}$, and the node $n_{j+s}$ from ${\bar{\cal A}_1}$ to ${\cal A}_1$,  leading to the new cut $\{{\cal A}_{j+1}, \bar{\cal A}_{j+1}\}$. There are $\sN$ cuts like this, and the above arguments hold for every new cut. Summing up the above inequality \eqref{eq: ConP-lgeqacut} over all these cuts, we have
     {
     \begin{align}
        \sum_{j=1}^{\sN} \sum_{i=1}^{d}\sum_{t=1}^{\sTe}(l_{j-i}^{(t)} + l_{j+s+i-1}^{(t)}) \geq \sum_{j=1}^{\sN} m_j \cdot s \cdot \sB_{2}, \label{eq: ConP-lgeqacutsum}
     \end{align}
     which is equivalent to
     \begin{align}
        2d \sum_{i=1}^{\sN} \sum_{t=1}^{\sTe} l_i^{(t)} \geq \sum_{j=1}^{\sN} m_j \cdot s \cdot \sB_{2}. \label{eq: ConP-lgeqacutsum2}
     \end{align}
        Note that \eqref{eq: ConP-lgeqacutsum} holds for $({\cal A}_1, \bar{\cal A}_1, \ldots, {\cal A}_{\sN}, \bar{\cal A}_{\sN})$ of arbitrary size $|{\cal A}_1| = \ldots = |{\cal A}_{\sN}| = s$. Recall that $m_j$ denotes the number of files that are not cached by nodes in ${\cal A}_j$. We have
        \begin{align}
            \sum_{j=1}^{\sN} |{\cal{M}}_{{\cal A}_j}| \leq \sum_{j=1}^{\sN}\sum_{u\in{\cal A}_j} |{\cal M}_u| \leq s\sum_{u=1}^{\sN} |{\cal M}_u| = sr\sN,
        \end{align}
        which implies that
        \begin{align} \label{eq: MissingIVCount}           
                \sum_{j=1}^{\sN} m_j = \sum_{j=1}^{\sN} (\sN - |{\cal{M}}_{{\cal A}_j}|) \geq \sN^2 - sr\sN.
        \end{align}
        Then, we have
        \begin{align}
                \sum_{i=1}^{\sN} \sum_{t=1}^{\sTe} l_i^{(t)} / (\sN\sB_{2}) 
                &\geq \max_{s\in\{1, \ldots \sN\}} \frac{\sum_{j=1}^{\sN} m_j \cdot s \cdot \sB_{2}}{2d \cdot \sN\sB_{2}} \nonumber\\
                &\geq \max_{s\in\{1, \ldots \sN\}} \frac{s(\sN- sr)}{2d}.
        \end{align}
        \begin{itemize}
                \item For the arbitrary file placement given computation load $r$, the optimizer is typically around
                $s = \frac{\sN}{2r}$ (requiring $s \geq 1$, i.e, $r \leq \frac{\sN}{2}$), giving the scaling of lower bound as 
                \begin{align}
                        \sum_{i=1}^{\sN} \sum_{t=1}^{\sTe} l_i^{(t)} / (\sN\sB_{2}) \gtrsim \frac{\sN^2}{8dr}.
                \end{align}
                \item For the cyclic file placement given computation load $r$, we have $m_j = \sN - s - r_j + 1$ for all $j\in[\sN]$, which refines the \eqref{eq: MissingIVCount} as 
                \begin{align}
                        \sum_{j=1}^{\sN} m_j = \sum_{i=1}^{\sN}(\sN - s - r_i + 1) = \sN^2 -s\sN - r\sN + \sN,
                \end{align}
                where $r_i$ is the number of nodes that cache file $w_i$. Then, we have
                \begin{align}
                        \sum_{i=1}^{\sN} \sum_{t=1}^{\sTe} l_i^{(t)} / (\sN\sB_{2}) \geq \max_{s\in\{1, \ldots \sN\}} \frac{s(\sN - s - r + 1)}{2d},
                \end{align}
                which is optimized at $s = \frac{\sN - r + 1}{2}$ (requiring $s \geq 1$, i.e, $r \leq \sN - 1$), giving the scaling of lower bound as
                \begin{align}
                        \sum_{i=1}^{\sN} \sum_{t=1}^{\sTe} l_i^{(t)} / (\sN\sB_{2}) \gtrsim \frac{(\sN - r + 1)^2}{8d}.
                \end{align}
        \end{itemize}
     }
        This completes the proof.

     \section{ The optimal NCL for All-to-all When $r\geq \frac{\sN}{2}$ and $d=1$}\label{sec: allunicast-arb}
    If the file placement can be arbitrarily designed, the optimal NCL can be achieved when $r\geq \lceil\frac{\sN}{2}\rceil$ and $d=1$ as:
    \begin{align}
        %     \sT_{2}(r,1) = \left\lceil \frac{\sN-r}{2} \right\rceil.
            \sT_{2}(r,1) = \frac{\sN-r}{2} .
    \end{align}
    For integer-valued computation load $r\in \{\lceil\frac{\sN}{2}\rceil, \ldots, \sN\}$, a file placement of the node $n_i$ as 
    \begin{align}
            \cM_i =& \{w_{i+4(t-1)}: t\in[\lceil\frac{\sN}{4}\rceil]\} \nonumber\\
            &\cup \{w_{i+1+4(t-1)}: t\in[\min((r-\lceil\frac{\sN}{4}\rceil),\lceil\frac{\sN}{4}\rceil)]\} \nonumber\\
            &\cup \{w_{i+2+4(t-1)}: t\in[\min((r-\lceil\frac{\sN}{2}\rceil),\lceil\frac{\sN}{4}\rceil)]\} \nonumber\\
            &\cup \{w_{i+3+4(t-1)}: t\in[r-\lceil\frac{3\sN}{4}\rceil]\},
    \end{align}
    achieve the NCL of $\left\lceil \frac{\sN-r}{2} \right\rceil.$ 
    The proposed transmission strategy for the problem is straightforward. We provide an example to illustrate it. 
    The file placement for $\sN=8$, $r=4$ is shown as Table~\ref{tab:arbseg-N8r4d1}. The symbol $*$ at the row $n_i$ and column $w_j$ where $i\in[8]$ and $j\in[8]$ means that the node $n_i$ compute the IVs from $w_j$. Node $n_i$ broadcast $X_{i}^{(1)} = v_i^{i+1}\oplus v_{i+1}^{i-1}$ and $X_{i}^{(2)} = v_{i+4}^{i+1} \oplus v_{i+5}^{i-1}$. For example, node $n_4$ broadcasts
        $$ X_{4}^{(1)} = v_4^{5}\oplus v_{5}^{3} \text{ and } X_{4}^{(2)} = v_{8}^{5} \oplus v_{1}^{3},$$ 
    and receives 
        \begin{align*}
            Y_4^{(1)}  = \left(
                    X_{3}^{(1)} = v_{3}^{4} \oplus v_{4}^{2},
                    X_{5}^{(1)} = v_{5}^{6} \oplus v_{6}^{4} \right)  
                \text{ and } Y_4^{(2)} = \left( X_{3}^{(2)} = v_{7}^{4} \oplus v_{8}^{2} 
                    X_{5}^{(2)} = v_{1}^{6} \oplus v_{2}^{4}\right) . 
    \end{align*}
    We can observe that node $n_4$ can get desired packets $\{v_3^4, v_7^4\}$ and $\{v_6^4, v_2^4\}$ from nodes $n_3$ and $n_5$, respectively. The other nodes can similarly obtain desired packets. It is clear that the NCL is $\sT_{2}(4,1) = 2 =\left\lceil \frac{8-4}{2} \right\rceil$.
    \begin{table}[tbp]
            \caption{A file placement achieving optimal transmission for $\sN=8$, $r=4$ and $d=1$ \label{tab:arbseg-N8r4d1}}
            \centering
            \begin{tabular}{|c|c|c|c|c|c|c|c|c|}
            \hline
              & $w_1$ & $w_2$ & $w_3$ & $w_4$ & $w_5$ & $w_6$ & $w_7$ & $w_8$ \\
            \hline
            $n_1$ & $*$ & $*$ &   &   & * & * &   &   \\
            \hline
            $n_2$ &   & $*$ & $*$ &   &   & * & * &   \\
            \hline
            $n_3$ &   &   & $*$ & $*$  &   &   & * & *  \\
            \hline
            $n_4$ & * &   &   & $*$ & $*$  &   &   & * \\
            \hline
            $n_5$ & * & * &   &   & $*$ & $*$  &   &   \\
            \hline
            $n_6$ &   & * & * &   &   & $*$ & $*$  &   \\
            \hline
            $n_7$ &   &   & * & * &   &   & $*$ & $*$  \\
            \hline
            $n_8$ & $*$  &   &  & * & * &   &   &  $*$ \\
            \hline
            \end{tabular}
    \end{table}

    The proof of the lower bound is presented as follows.
    From the definition of the mutual information and Fano's inequality, we have
        \begin{align}
                H(v_{[\sN]}^{i} | {\cal M}_i) &= I(v_{[\sN]}^{i}; Y_{i}^{\sTe}|{\cal M}_i) + H(_{[\sN]}^{i}, Y_{i}^{\sTe} | {\cal M}_i) \nonumber \\
                &= H( Y_{i}^{\sTe}|{\cal M}_i) - H( Y_{i}^{\sTe}|v_{[\sN]}^{i}, {\cal M}_i) + \sTe\epsilon_{\sTe}\nonumber\\
                &\leq H(Y_{i}^{\sTe}),
        \end{align}
        where $\epsilon_{\sTe}$ vanishes as $\sTe\to\infty$. This implies that
        \begin{align}
                H(Y_{i}^{\sTe}) &\geq H(v_{[\sN]}^{i} | {\cal M}_i) - \sTe\epsilon_{\sTe} \nonumber\\
                & = (\sN-|{\cal M}_i|)\sB_{2} - \sTe\epsilon_{\sTe}.
        \end{align}
        Summing the above inequality over all $i\in[\sN]$ yields
        \begin{align}
                2\sum_{i}^{\sN} \sum_{t=1}^{\sTe} l_i^{(t)}
                &= \sum_{i=1}^{\sN} H(Y_{i}^{\sTe}) \nonumber\\
                &\geq \sum_{i=1}^{\sN} (\sN-|{\cal M}_i|)\sB_{2} - \epsilon' \nonumber\\
                & = \sN(\sN-r)\sB_{2}- \epsilon',
        \end{align}
        i.e., 
        \begin{align}
                \sT^{\star}_{2}(r,1) = \sum_{i}^{\sN} \sum_{t=1}^{\sTe} l_i^{(t)} / (\sN\sB_{2}) &\geq \frac{\sN-r}{2} - \epsilon',
        \end{align}
        where $\epsilon^{\prime}$ vanishes as $\sTe\to\infty$.
        This completes the proof.

\end{appendices}

% {\appendices\limits_{k\in\cal{T}}
% \section*{Proof of the First Zonklar Equation}
% \section*{Proof of the Second Zonklar Equation}
% Appendix two text goes here.}

% \section{References Section}
% You can use a bibliography generated by BibTeX as a .bbl file.
%  BibTeX documentation can be easily obtained at:
%  http://mirror.ctan.org/biblio/bibtex/contrib/doc/
%  The IEEEtran BibTeX style support page is:
%  http://www.michaelshell.org/tex/ieeetran/bibtex/
 

\bibliographystyle{IEEEtran}
\bibliography{ref}

\end{document}